\documentclass[a4paper]{scrartcl}

\usepackage[usenames,dvipsnames,table]{xcolor}
\usepackage{graphicx}
\usepackage{float}

\usepackage{amsmath,amssymb,amsthm}
\usepackage{xspace}

\usepackage{fullpage}

\usepackage{tikz}
\usetikzlibrary{calc,positioning,arrows}
\usepackage{paralist}
\usepackage[ruled,linesnumbered] 
       {algorithm2e} 
\SetAlgoInsideSkip{medskip}

\usepackage[most]{tcolorbox}

\newtheorem{theorem}{Theorem}
\newtheorem{lemma}{Lemma}
\newtheorem{proposition}{Proposition}
\newtheorem{corollary}{Corollary}

\newtheorem{claim}{Claim}

\theoremstyle{plain}
\newtheorem{rrule}{Reduction Rule}

\theoremstyle{remark}
\newtheorem{observation}{Observation}

\newcommand{\C}{\mathcal{C}}
\newcommand{\F}{\mathcal{F}}
\newcommand{\N}{\mathbb{N}}

\newcommand{\NP}{\ensuremath{\mathsf{NP}}\xspace}

\newcommand{\containment}{\ensuremath{\mathsf{NP\subseteq coNP/poly}}\xspace}
\newcommand{\mktwo}{\ensuremath{\mathsf{MK[2]}}\xspace}

\newcommand{\Oh}{\mathcal{O}}
\newcommand{\T}{\mathcal{T}}

\newcommand{\dunion}{\mathbin{\dot{\cup}}}

\def\3SAT{\textup{\textsc{3sat}}}
\def\NAE3SAT{\textup{\textsc{nae 3sat}}}
\def\1IN3SAT{\textup{\textsc{1-in-3sat}}}
\def\P1IN3SAT{\textsc{positive 1-in-3sat}}
\def\pq3SAT{\textup{\textsc{$(p,q)$-3sat}}}
\def\2,23SAT{\textup{\textsc{$(2,2)$-3SAT}}}

\def\SCS{\textsc{Stable Cutset}}

\newcommand{\probname}[1]{\textup{\textsc{#1}}\xspace}
\newcommand{\cnfsatn}{\probname{CNF-SAT[$n$]}}
\newcommand{\hittingset}{\probname{Hitting Set}}
\newcommand{\hittingsetn}{\probname{Hitting Set[$n$]}}

\newcommand{\multicoloredhittingsetn}{\probname{Multicolored Hitting Set[$n$]}}
\newcommand{\setsplitting}{\probname{Set Splitting}}
\newcommand{\setsplittingn}{\probname{Set Splitting[$n$]}}
\newcommand{\stablecutset}{\probname{Stable Cutset}}
\newcommand{\stablecutsetk}{\probname{Stable Cutset[$k$]}}
\newcommand{\stablecutsetmodclique}{\probname{Stable Cutset[mod clique]}}
\newcommand{\stablecutsetmodcluster}{\probname{Stable Cutset[mod cluster]}}
\newcommand{\stablecutsetmodcocluster}{\probname{Stable Cutset[mod co-cluster]}}
\newcommand{\stablecutsetmodlinforest}{\probname{Stable Cutset[mod lin forest]}}
\newcommand{\stablecutsetmodpath}{\probname{Stable Cutset[mod path]}}
\newcommand{\stablecutsettc}{\probname{Stable Cutset[tc]}}
\newcommand{\stablecutsettw}{\probname{Stable Cutset[tw]}}
\newcommand{\stablecutsetvc}{\probname{Stable Cutset[vc]}}
\newcommand{\threecoloring}{\probname{$3$-Coloring}}

\newtcolorbox{problembox}[2][]{%
  attach boxed title to top left
               = {yshift=-8pt,xshift = 8pt},
  colback      = black!5,
  colframe     = black,
  coltitle     = black, 
  fonttitle    = \textcolor{black},
  colbacktitle = black!15,
  title        = #2,#1,
  enhanced,
  before skip = 8pt, 
}

\newcommand{\introduceproblem}[3]{
\begin{problembox}{#1}
\vspace{0.2cm}
\begin{tabular}{l l}
\emph{Instance:}& #2\\
\emph{Question:}& #3
\end{tabular}
\end{problembox}}

\newcommand{\introduceparameterizedproblem}[4]{
\begin{problembox}{#1}
\vspace{0.2cm}
\begin{tabular}{l l}
\emph{Instance:}& #2\\
\emph{Parameter:}& #3\\
\emph{Question:}& #4
\end{tabular}
\end{problembox}}

\definecolor{teal}{rgb}{0.00,0.5,0.5} 
\definecolor{flax}{rgb}{0.90,0.78,0.00} 
\colorlet{lightred}{red!25}
\colorlet{lightblue}{blue!25}

\title{On polynomial kernelization for Stable Cutset}

\author{Stefan Kratsch\\Humboldt-Universität zu Berlin, Berlin, Germany\\kratsch@informatik.hu-berlin.de
\and
Van Bang Le\\Institut für Informatik, Universität Rostock, Rostock, Germany\\van-bang.le@uni-rostock.de
}

\begin{document}

\maketitle              

\centerline{\emph{For Dieter Kratsch on his 65th birthday}}
\bigskip

\begin{abstract}
A \emph{stable cutset} in a graph $G$ is a set $S\subseteq V(G)$ such that vertices of $S$ are pairwise non-adjacent and such that $G-S$ is disconnected, i.e., it is both stable (or independent) set and a cutset (or separator). Unlike general cutsets, it is $\NP$-complete to determine whether a given graph $G$ has any stable cutset. Recently, Rauch et al.\ [FCT 2023] gave a number of fixed-parameter tractable (FPT) algorithms, time $f(k)\cdot |V(G)|^c$, for \stablecutset under a variety of parameters $k$ such as the size of a (given) dominating set, the size of an odd cycle transversal, or the deletion distance to $P_5$-free graphs. Earlier works imply FPT algorithms relative to clique-width and relative to solution size.

We complement these findings by giving the first results on the existence of polynomial kernelizations for \stablecutset, i.e., efficient preprocessing algorithms that return an equivalent instance of size polynomial in the parameter value. Under the standard assumption that $\mathsf{NP\nsubseteq coNP/poly}$, we show that no polynomial kernelization is possible relative to the deletion distance to a single path, generalizing deletion distance to various graph classes, nor by the size of a (given) dominating set. We also show that under the same assumption no polynomial kernelization is possible relative to solution size, i.e., given $(G,k)$ answering whether there is a stable cutset of size at most $k$. On the positive side, we show polynomial kernelizations for parameterization by modulators to a single clique, to a cluster or a co-cluster graph, and by twin cover.
%
\end{abstract}
\section{Introduction}
A \emph{stable cutset} in a graph $G$ is a set $S\subseteq V(G)$ such that vertices of $S$ are pairwise non-adjacent and such that $G-S$ is disconnected, i.e., it is both stable (or independent) set and a cutset (or separator). The \stablecutset problem asks whether or not a given graph has a stable cutset:

\introduceproblem{\stablecutset}{A graph $G$.}{Does $G$ have a stable cutset?}

Unlike finding general cutsets, the \stablecutset problem is $\NP$-complete, and this remains true even on various restricted graph classes~\cite{Chvatal84,KleinF96,BrandstadtDLS00,LeR03,LeMM08}. Moreover, it follows from a result by Brandstädt et al.~\cite{BrandstadtDLS00} that, assuming the Exponential-Time Hypothesis, \SCS\ cannot be solved in $2^{o(n)}$ time on $n$-vertex graphs, even when restricted to $K_4$-free graphs of diameter~4.

The importance of stable cutsets has been demonstrated first in connection to perfect graphs~\cite{Tucker83,CorneilF93}. Tucker~\cite{Tucker83} proved that if a graph $G$ has a stable cutset $S$ such that no induced cycle of odd length at least five has a vertex in $S$, then the coloring problem on $G$ reduces to the same problem on the smaller graphs induced by $S$ and the connected components of $G-S$. Later, structural and algorithmic aspects of stable cutsets were studied: It is known that any $n$-vertex graph with at most $2n-4$ edges has a stable cutset and such one can be computed in polynomial time~\cite{ChenY02}, extremal graphs with $2n-3$ edges and without stable cutsets have been described in~\cite{LeP13}, and graphs with small stable cutsets were studied in~\cite{ChenFJ02}. 
While the complexity of \SCS\ in graphs of maximum degree four is still open, it is known that \SCS\ remains $\NP$-complete in planar graphs of maximum degree five \cite{LeMM08} and in 5-regular line graphs of bipartite graphs \cite{LeR03}.

In fact, stable cutsets in \emph{line graphs} have been studied under the notion of a matching cut. Here, a \emph{matching cut} is an edge cut that is also a matching. It can be seen that, for graphs $G$ with minimum degree at least two, $G$ has a matching cut if and only if the line graph of~$G$, $L(G)$, admits a stable cutset. 
Graphs having a matching cut were first discussed in \cite{Graham70,Chvatal84}. It follows from~\cite{Chvatal84,Moshi89} that \SCS\ is $\NP$-complete even on line graphs of bipartite graphs, hence on perfect graphs, and it follows from~\cite{LuckePR23} that computing a \emph{maximum} (\emph{minimum}) stable cutset in line graphs is $\NP$-hard. 
We refer to the recent papers \cite{KomusiewiczKL20,ChenHLLP21,GolovachKKL22,FeghaliLPR23isaac,LuckePR23} for more information on matching cuts, i.e., stable cutsets in line graphs.

The (likely) intractability of \stablecutset on general graphs (and many restricted classes) as well as its relation to graph structure motivate a study of its parameterized complexity: How does structure of the input graph $G$ influence the complexity of \SCS? E.g., is the problem significantly easier to solve not just on, say, chordal graphs, but also on graphs that are close to being chordal (e.g., at most $\ell$ vertex deletions away)? Similarly, does having (small) treewidth at most $w$ or asking for a stable set of (small) size at most $k$ allow for faster algorithms? To this end, one studies \stablecutset with any chosen parameter(s), say, $k$, and asks whether the resulting parameterized problem is \emph{fixed-parameter tractable} (FPT), i.e., if there is an algorithm that decides inputs $(x,k)$ in time $f(k)\cdot |x|^c$ for some computable function $f$ (also called an FPT-algorithm). Some but (likely) not all FPT problems also admit a strong form of efficient preprocessing: A \emph{polynomial kernelization} is an efficient algorithm that given $(x,k)$ returns an input $(x',k')$ with same answer but size bounded by a polynomial in $k$.

So far, there are only few results regarding the parameterized complexity of \stablecutset, in particular, no dedicated results on kernelization. A general result of Marx et al.~\cite{MarxOR13} implies that \stablecutsetk (\SCS\ parameterized by solution size) is fixed-parameter tractable. Very recently, Rauch et al.~\cite{RauchRS23} proved, among others, that \SCS\ is FPT when parameterized by the domination number, by the distance to bipartite graphs, and by the distance to chordal graphs. Bodlaender observed that the problem can be solved in $\Oh(3^wn^c)$ when a tree decomposition of width $w$ is given.\footnote{Private communication in August 2015.}

\paragraph{Our work.}
We initiate the study of upper and lower bounds for kernelizations for the \stablecutset problem. The results obtained have some similarity to what is known for \threecoloring, in particular, results obtained by Jansen and Kratsch~\cite{JansenK13}. While the problems are different, some similarities can be observed:
A set $S\subseteq V(G)$ is a stable cutset of $G$ if and only if $G-S$ is disconnected, and hence can be partitioned into two parts $G[A]$ and $G[B]$ with no connecting edges and with $A,B\neq\emptyset$. In other words, $G$ has a stable cutset if and only if its vertices can be partitioned into $V(G)=S\dunion A\dunion B$ such that $S$ is a (possibly empty) stable set, $A$ and $B$ are both nonempty, and there is no edge with one endpoint in $A$ and one in $B$. Intuitively, this leaves us with a $3$-coloring-like problem (that could also be cast as $H$-coloring with vertices $S$, $A$, and $B$, suitable edges/loops, and condition that $A,B\neq\emptyset$). A key difference between \stablecutset and \threecoloring, however, is that a stable cutset may be rather local, in the sense of separating only a small part, say, $A$ and having most of the graph in $B$. Unlike \threecoloring, the \SCS\ problem is also not monotone, i.e., deletion of vertices may turn the instance from yes to no.

We obtain several several polynomial kernelizations for \stablecutset with respect to different parameters. All but one of the results involve one or more marking-based reduction rules that preserve one vertex/component each for certain adjacencies of $G-X$ into $X$, similar to what was done for \threecoloring parameterized by vertex cover~\cite{JansenK13}.

\begin{theorem}\label{theorem:intro:upperbounds}
    Kernelizations with the following bounds exist for \stablecutset parameterized by $|X|$ when a set $X$ with the required properties is provided:
    \begin{itemize}
        \item $\Oh(|X|)$ vertices when $G-X$ is a clique.
        \item $\Oh(|X|^3)$ vertices when $X$ is a vertex cover (or $G-X$ is a stable set).
        \item $\Oh(|X|^3)$ vertices when $X$ is a twin cover.
        \item $\Oh(|X|^5)$ vertices when $G-X$ is a cluster graph.
        \item $\Oh(|X|^3)$ vertices when $G-X$ is a co-cluster graph.
    \end{itemize}
\end{theorem}

As is common, we assume that a suitable witness (in this case a suitable set $X$) is provided with the input. That being said, all parameterizations in the previous theorem have efficient constant-factor approximation algorithms.

Our main lower bound result is to show that \stablecutset parameterized by the size of a deletion set (also called modulator) to a single path is hard for a class called \mktwo (introduced by Hermelin et al.~\cite{HermelinKSWW15}) under polynomial parameter transformations (a special parameterized reduction~\cite{BodlaenderTY11}). To this end, we first establish that \setsplittingn\footnote{Given a ground set $U$ and a set family $\F$ over $U$, is there a subset $S\subseteq U$ that contains neither $F$ nor the complement $\overline{F}$ for any $F\in\F$. The parameter is $n=|U|$.} is complete for \mktwo. We then give the required reduction from \setsplittingn to \stablecutsetmodlinforest, which also works for parameterization by (connected) domination number. It was proved by Hermelin et al.~\cite{HermelinKSWW15} that \mktwo-hard problems do not admit polynomial kernelizations unless \containment. They also conjectured that such problems do not admit so-called polynomial \emph{Turing kernelizations}.

\begin{theorem}\label{theorem:intro:lowerbounds}
    The following parameterizations of \stablecutset are hard for the class \mktwo under polynomial parameter transformations and thus do not admit polynomial kernelizations unless \containment (and PH collapses):
    \begin{itemize}
        \item \SCS\ parameterized by the size of $X$ such that $G-X$ is a single path.
        \item \SCS\ parameterized by the size of a given (connected) dominating set $X$. 
    \end{itemize}
\end{theorem}

Furthermore, we show a simple construction that, unless \containment, rules out polynomial kernelizations for \stablecutset parameterized by solution size or by treewidth (or by vertex integrity).

\paragraph{Organization.}
Section~\ref{section:preliminaries} contains the preliminaries. In Section~\ref{section:properties} we collect various properties of stable cutsets and develop some first reduction rules. Section~\ref{section:upperbounds} contains our polynomial kernelizations, while Section~\ref{section:lowerbounds} contains the claimed lower bounds for kernelization. We conclude in Section~\ref{section:conclusion}.

\section{Preliminaries}\label{section:preliminaries}

\paragraph{Graphs.} 
We consider only finite and simple graphs $G=(V,E)$ with vertex set $V(G)=V$ and edge set $E(G)=E$. Throughout we use $n=|V|$ and $m=|E|$. A set $S\subseteq V(G)$ is a \emph{stable set} in $G$ if $E(G)$ contains no edges between vertices of $S$, respectively a \emph{clique} if all vertices in $S$ are pairwise adjacent. A set $S\subseteq V(G)$ is a \emph{cutset} of $G$ if $G-S$ is disconnected; the empty set is a cutset of any disconnected graph $G$. A \emph{consistent partition} of a disconnected graph $G$ is a partition $V(G)=A\dunion B$ with $A,B\neq\emptyset$ and such that no edge of $G$ has one endpoint in $A$ and the other in $B$. 

\begin{proposition}
    A set $S\subseteq V(G)$ is a cutset of $G$ if and only if $G-S$ has a consistent partition.
\end{proposition}

We use $N(v)$ to denote the (open) neighborhood of a vertex $v$ and, $N[v]=N(v)\cup\{v\}$ for its closed neighborhood. A vertex $v$ is \emph{simplicial} if $N(v)$ is a clique. A set $M\subseteq V(G)$ is a \emph{module} of $G$ if $N(u)\setminus M= N(v)\setminus M$ holds for all $u,v\in M$.
Two vertices $u$ and $v$ of $G$ are \emph{true twins} if $N[u]=N[v]$ and \emph{false twins} if $N(u)=N(v)$, i.e., if they have the same neighbors outside of $\{u,v\}$ and they are either adjacent or non-adjacent. Say that $u$ and $v$ are \emph{twins} if they are true twins or false twins.

\begin{proposition}
    Vertices being twins in a graph $G$ is an equivalence relation on $V(G)$ and each equivalence class is a clique or stable set. Moreover, each equivalence class is a (clique or stable) module of $G$.
\end{proposition}

\paragraph{Parameterized complexity and kernelization.}
A \emph{parameterized problem} is a language $L\subseteq\Sigma^*\times\N$; the second component $k$ of an instance $(x,k)$ is called its \emph{parameter}. A parameterized problem $L$ is \emph{fixed-parameter tractable} if there is an algorithm $A$, a computable function $f$, and a constant $c$ such that $A$ decides $(x,k)\in L$ in time $f(k)|x|^c$. A \emph{kernelization} for a parameterized problem $L$ is an efficient algorithm that given $(x,k)\in\Sigma^*\times\N$ takes time polynomial in $|x|+k$ and returns an instance $(x',k')$ of size $|x'|+k'$ upper bounded by $h(k)$ for some computable function $h\colon\N\to\N$ and with $(x,k)\in L$ if and only if $(x',k')\in L$. The function $h$ is also called the \emph{size} of the kernelization and for a \emph{polynomial kernelization} it is required that $h$ is polynomially bounded. It is common to allow for a kernelization to also (correctly) answer yes or no regarding membership of $(x,k)$ in $L$, which could be handled by returning any constant-size yes- or no-instance. A \emph{polynomial parameter transformation} from $Q\subseteq\Sigma^*\times\N$ to $Q'\subseteq\Gamma^*\times\N$ is an efficiently computable mapping $\pi\colon\Sigma^*\times\N\to\Gamma^*\times\N$ such that for each $(x,k)\in\Sigma^*\times\N$ and $(x',k')=\pi((x,k))$ we have that $(x,k)\in Q$ if and only if $(x',k')\in Q'$ and that $k'$ is polynomially bounded in $k$.

\paragraph{Graph parameters.} Let $G=(V,E)$ be a graph. A \emph{vertex cover} of $G$ is a set $X\subseteq V$ such that $G-X$ is a stable set. A \emph{twin cover} of $G$ is a set $X\subseteq V$ such that $N[u]=N[v]$ holds for every edge $\{u,v\}$ of $G-X$, i.e., every two adjacent vertices of $G-X$ are twins (which means that they must be true twins). For a graph class $\C$, a \emph{modulator of $G$ to $\C$} is a set $X\subseteq V$ such that $G-X\in\C$. E.g., the vertex covers are exactly the modulators to the class of stable sets, and each twin cover is also a modulator to the class of cluster graphs (but not vice versa). 
A \emph{dominating set} of $G$ is a set $X\subseteq V$ such that $V(G)=N[X]$, i.e., such that each vertex of $G$ is contained in $X$ or adjacent to a vertex of $X$.

\section{Some properties of stable cutsets}\label{section:properties}

In this section, we observe and recall several properties of stable cutsets. This leads to a number of reduction rules that will be used in Section~\ref{section:upperbounds}.

The first three rules handle several cases where $G$ has a trivial stable cutset, such as the empty set for disconnected graphs, or no (stable) cutset, such as for cliques. These can be applied for any parameterization because they outright answer the problem rather than modifying the input. Throughout we assume that reduction rules are applied in the order that they are introduced in, so that later rules benefit from structural properties that are established by earlier ones.

\begin{rrule}\label{rrule:disconnected:cutvertex}
    If $G$ is disconnected or $G$ has a cut-vertex then return yes.
\end{rrule}

\begin{rrule}\label{rrule:clique}
    If $G$ is a clique then return no.
\end{rrule}

\begin{rrule}\label{rrule:stableneighborhood}
    If $G$ has any vertex whose neighborhood is a stable set then return yes.
\end{rrule}

\begin{lemma}\label{lemma:rrule:stableneighborhood:safe}
    Reduction Rule~\ref{rrule:stableneighborhood} is safe, i.e., if Reduction Rules~\ref{rrule:disconnected:cutvertex} and~\ref{rrule:clique} do not apply to $G$ and the neighborhood of $v$ is a stable set then $G$ has a stable cutset.
\end{lemma}

\begin{proof}
    We need to check that the few corner cases without stable cutset are correctly handled by Reduction Rule~\ref{rrule:clique}. 
    As Reduction Rules~\ref{rrule:disconnected:cutvertex} and \ref{rrule:clique} do not apply to $G$, it must have at least three vertices. If $V(G)=N[v]$ then $v$ would be a cut vertex (separating at least two neighbors of $v$ as $N(v)$ is a stable set) and Reduction Rule~\ref{rrule:disconnected:cutvertex} would apply; a contradiction. Else, let $u\in V(G)\setminus N[v]$ and observe that $S=N(v)$ is a stable cutset with consistent partition $V(G)\setminus S=A\dunion B$ where $A=\{v\}\neq\emptyset$ and $B=V(G)\setminus N[v]\supseteq\{u\}\neq\emptyset$.
\end{proof}

Several of our reduction rules allow the deletion of a single vertex $v$ or of a clique $C$ from $G$.
The following lemma eases proofs of safeness for the case of vertex deletion by showing that this cannot change the answer from yes to no.

\begin{lemma}\label{lemma:deletevertex:nofalsenegative}
    Let $G$ be a graph to which Reduction Rules~\ref{rrule:disconnected:cutvertex}, \ref{rrule:clique}, and~\ref{rrule:stableneighborhood} do not apply and let $v$ be any vertex of $G$. If $G$ has a stable cutset then so does $G-v$.
\end{lemma}

\begin{proof}
    Let $v\in V(G)$, let $S$ be a stable cutset of $G$, and let $V\setminus S=A\dunion B$ be a consistent partition of $G-S$. If $v\in S$ then $S\setminus\{v\}$ is a stable cutset of $G-v$ with consistent partition $V(G-v)\setminus S= A\dunion B$. Else, w.l.o.g.\ $v\in A$. If $A=\{v\}$ then $N(v)\subseteq S$ is a stable set, contradicting the assumption that Reduction Rule~\ref{rrule:stableneighborhood} does not apply to $G$. Thus $A\supsetneq\{v\}$ but then $S$ is a stable cutset of $G-v$ with consistent partition $V(G-v)\setminus S=(A\setminus\{v\})\dunion B$.
\end{proof}

Something similar can be done for the case of deleting a clique $C$ with $|C|\geq 2$ from a graph $G$. Crucially, any stable cutset $S$ of $G$ contains at most one vertex of $C$ while the rest $C\setminus S$ must be part of a single connected component of $G-S$. For a clique $C$, we say that a stable cutset $S$ is \emph{$C$-simple} if $S=N(C)$ or if $S=N(C\setminus \{v\})$ for some vertex $v\in C$. For a given clique $C$, it can be tested efficiently if there is any $C$-simple stable cutset and, if so, we may return yes. (There is probably no efficient way of doing this for all cliques of $G$, so it has to be done ``on demand'' when a suitable clique is found by a reduction rule.) Deleting any clique $C$ without $C$-simple stable cutset cannot change the answer from yes to no (analogous to Lemma~\ref{lemma:deletevertex:nofalsenegative}).

\begin{lemma}\label{lemma:deleteclique:nofalsenegative}
    Let $G$ be a graph and let $C\subseteq V(G)$ be a clique in $G$ without $C$-simple stable cutset. If $G$ has a stable cutset then so does $G-C$.
\end{lemma}

\begin{proof}
    Let $S$ be a minimal stable cutset of $G$, and let $V\setminus S=A\dunion B$ be a consistent partition of $G-S$. Since $C$ is a clique, $C\subseteq A\cup S$ or $C\subseteq B\cup S$. By symmetry, say $C\subseteq A\cup S$. If $A\setminus C\not=\emptyset$ then $S\setminus C$ is a stable cutset of $G-C$ and we are done. If not, $A=C\setminus S$, hence 
    by the minimality of $S$, $S=N(C)$ (if $C\cap S=\emptyset$) or $S=N(C\setminus\{v\})$ for the vertex $v$ in $C\cap S$. But then $S$ is a $C$-simple stable cutset, a contradiction.
\end{proof}

We can now present several further reduction rules that allow to delete vertices from the graph under certain conditions. They can be applied for any parameterization that is monotone under vertex deletions such as deletion distance to some hereditary graph class but not, e.g., for max leaf number.

\begin{rrule}\label{rrule:simplicial}
    Delete any simplicial vertex from $G$.
\end{rrule}

\begin{lemma}
    Reduction Rule~\ref{rrule:simplicial} is safe, i.e., if Reduction Rules~\ref{rrule:disconnected:cutvertex}, \ref{rrule:clique}, and \ref{rrule:stableneighborhood} do not apply to $G$ and if $v$ is a simplicial vertex in $G$ then $G$ has a stable cutset if and only if $G-v$ has a stable cutset.
\end{lemma}

\begin{proof}
    By Lemma~\ref{lemma:deletevertex:nofalsenegative}, if $G$ has a stable cutset then $G-v$ also has a stable cutset. It remains to check the converse.

    Let $S$ be a stable cutset of $G-v$ and let $V(G-v)\setminus S=A\dunion B$ be a consistent partition of $(G-v)-S$. Since Reduction Rule~\ref{rrule:stableneighborhood} does, in particular, not apply to $v$, its neighborhood is not a stable set, i.e., $C=N_G(v)$ is a clique of size at least two. The clique $C$ is also present in $G-v$ and at least one of $A$ and $B$ cannot contain vertices of $C$. W.l.o.g., $C\subseteq S\cup A$. As all neighbors of $v$ are in $C$, in particular there is no neighbor in $B$, this implies that $S$ is also a stable cutset of $G$ with consistent partition $V(G)\setminus S=(A\cup\{v\})\dunion B$.
\end{proof}

\begin{rrule}\label{rrule:comparableneighborhood}
    If $G$ has vertices $u$ and $v$ with $N(v)\subseteq N(u)$ then delete $v$. 
\end{rrule}

\begin{lemma}\label{lemma:rrule:comparableneighborhood:safe}
    Reduction Rule~\ref{rrule:comparableneighborhood} is safe, i.e., if Reduction Rules~\ref{rrule:disconnected:cutvertex}, \ref{rrule:clique}, and \ref{rrule:stableneighborhood} do not apply and $G$ has two vertices $u$ and $v$ with $N(v)\subseteq N(u)$ then $G$ has a stable cutset if and only if $G-v$ has a stable cutset.
\end{lemma}

\begin{proof}
    By Lemma~\ref{lemma:deletevertex:nofalsenegative}, if $G$ has a stable cutset then $G-v$ also has a stable cutset. It remains to check the converse.
    
    Let $S$ be a stable cutset of $G-v$ and let $V(G-v)\setminus S=A\dunion B$ be a consistent partition of $(G-v)-S$. Observe that $u$ and $v$ are not adjacent, as otherwise $u\in N_G(v)\subseteq N_G(u)$; thus $N_G(v)\subseteq N_G(u)=N_{G-v}(u)$. If $u\in S$ then $N_{G-v}(u)\subseteq A\cup B$ as $S$ is a stable set. It follows that $N_G(v)\subseteq N_{G-v}(u)\subseteq A\cup B$. Thus, $S\cup\{v\}$ is a stable cutset of $G$ with consistent partition $V(G)\setminus (S\cup\{v\})=A\dunion B$. If, instead, $u\in A$ then $N_{G-v}(u)\subseteq S\cup A$. It follows that $N_G(v)\subseteq N_{G-v}(u)\subseteq S\cup A$. Thus, $S$ is a stable cutset of $G$ with consistent partition $V(G)\setminus S=(A\cup\{v\})\dunion B$. The case that $u\in B$ is symmetric.
\end{proof}

Amongst others, Reduction Rule~\ref{rrule:comparableneighborhood} strongly restricts the size of stable modules of $G$, i.e., of modules that are also stable sets.

\begin{lemma}
    If Reduction Rules~\ref{rrule:disconnected:cutvertex}, \ref{rrule:clique}, \ref{rrule:stableneighborhood}, and \ref{rrule:comparableneighborhood} do not apply to $G$ then each module of $G$ that is a stable set has size at most one.
\end{lemma}

\begin{proof}
    Let $M$ be a module of $G$ that is also a stable set, and let $u$ and $v$ be two different vertices of $M$. By definition of a module we have $N(u)\setminus M=N(v)\setminus M$. Since $M$ is a stable set, it follows that $N(u)=N(u)\setminus M=N(v)\setminus M=N(v)$. But then Reduction Rule~\ref{rrule:comparableneighborhood} would apply to $u$ and $v$; a contradiction.
\end{proof}

For clique modules, i.e., modules that are also cliques, the following reduction rule can be applied.

\begin{rrule}\label{rrule:cliquemodule}
    If $Q$ is a clique module of $G$ with at least three vertices then delete any vertex $v\in Q$ from $G$.
\end{rrule}

\begin{lemma}\label{lemma:rrule:cliquemodule:safe}
    Reduction Rule~\ref{rrule:cliquemodule} is safe, i.e., if Reduction Rules~\ref{rrule:disconnected:cutvertex}, \ref{rrule:clique}, and \ref{rrule:stableneighborhood} do not apply and $G$ has a clique module $Q$ of size at least three then $G$ has a stable cutset if and only if $G-v$ has stable cutset, for each $v\in Q$.
\end{lemma}

\begin{proof}
    By Lemma~\ref{lemma:deletevertex:nofalsenegative}, if $G$ has a stable cutset then $G-v$ has a stable cutset. It remains to check the converse.

    Let $S$ be a stable cutset of $G-v$ and let $V(G-v)\setminus S=A\dunion B$ be a consistent partition.
    As $Q\setminus\{v\}$ is a clique of size at least two in $G-v$, there is at least one vertex $u\in Q\setminus \{v\}$ that is not in $S$. W.l.o.g., $u\in A$, hence $N_{G-v}(u)\subseteq A\cup S$. It follows that $S$ is a stable cutset for $G$ with consistent partition $V(G)\setminus S=(A\cup\{v\})\dunion B$: We have $N_G[v]=N_G[u]=\{v\}\cup N_{G-v}[u]\subseteq (A\cup\{v\})\cup S$.
\end{proof}

\begin{observation}\label{observation:rrule:cliquemodule:efficient}
    If $Q$ is a clique module of $G$ then each $Q'\subseteq Q$ is also a clique module. In particular, if $|Q|\geq 3$ then there is a clique module $Q'\subseteq Q$ of size exactly three to which Reduction Rule~\ref{rrule:cliquemodule} can be applied. Thus, exhaustively checking triplets of vertices for being a clique module (and possibly deleting one of them) suffices to efficiently implement Reduction Rule~\ref{rrule:cliquemodule}. Afterwards, each clique module has size at most two. (A more efficient routine can be obtained via linear-time modular decomposition.)
\end{observation}

\section{Polynomial kernelizations for Stable Cutset}\label{section:upperbounds}

\subsection{Vertex cover number and twin cover number}

We first give a kernelization with a cubic number of vertices for \stablecutset parameterized by the size of a given vertex cover. The approach is similar to a kernelization for \probname{$3$-Coloring} parameterized by vertex cover~\cite{JansenK13}. We then generalize the result to work also for the smaller parameter twin-cover.

\introduceparameterizedproblem{\stablecutsetvc}{A graph $G$ and a vertex cover $X$ of $G$.}{$|X|$.}{Does $G$ have a stable cutset?}

We use Reduction Rules~\ref{rrule:disconnected:cutvertex}, \ref{rrule:clique}, and~\ref{rrule:stableneighborhood} as well as the following reduction rule.

\begin{rrule}\label{rrule:vertexcover:marking}
   If $|V(G)\setminus X|>|X|^3$, then, starting with all vertices as unmarked, mark (isolated) vertices in $G-X$ according to the following condition:
    \begin{itemize}
        \item For each $(x_1,x_2,x_3)\in X^3$ mark a vertex $p$ that is adjacent to $x_1$, $x_2$, and~$x_3$.
    \end{itemize}
    Delete any unmarked vertex $v\in V(G)\setminus X$, i.e., return the instance $(G-v,X)$.
\end{rrule}

\begin{observation}
    If $|V(G)\setminus X|>|X|^3$ then there will always be an unmarked vertex, so each application will lead to the deletion of a vertex. (A more efficient implementation would of course delete all unmarked vertices in one go.)
\end{observation}

We now prove that Reduction Rule~\ref{rrule:vertexcover:marking} is safe. Part of the proof implicitly handles vertices of $V(G)\setminus X$ that have degree at most two, as they cannot have neighbors in $S\cap X$, $A\cap X$, and $B\cap X$ for any partition $V(G)=S\dunion A\dunion B$.

\begin{lemma}
    Reduction Rule~\ref{rrule:vertexcover:marking} is safe.
\end{lemma}

\begin{proof}
    Let $v$ be the vertex that is being deleted. By Lemma~\ref{lemma:deletevertex:nofalsenegative}, if $G$ has a stable cutset then so does $G-v$. Thus, if $(G,X)$ is yes then so is $(G-v,X)$. It remains to show the converse.

    Assume that $(G-v,X)$ is yes and let $S$ be a stable cutset of $G-v$ with consistent partition $V(G-v)\setminus S=A\dunion B$. If $v$ has no neighbor in $S$ then $S\cup\{v\}$ is a stable cutset of $G$ with same consistent partition. Similarly, if $v$ has no neighbor in $A$ (respectively in $B$) then $S$ is a stable cutset of $G$ with consistent partition $V(G)\setminus S=A\dunion (B\cup\{v\})$ (resp.~$V(G)\setminus S=(A\cup\{v\})\dunion B$). It remains to consider the case that all three types of neighbors exist. 

    Note that $N_G(v)\subseteq X$ and let $x_1\in N_G(v)\cap S$, let $x_2\in N_G(v)\cap A$, and let $x_3\in N_G(v)\cap B$. As $v$ was unmarked, there is a marked vertex $u\in V(G-v)\setminus X$ that is adjacent to $x_1\in S$, $x_2\in A$, and $x_3\in B$. Because of $x_1$, we must have $u\notin S$. Similarly, because of $x_2$ and $x_3$ we must have $u\notin B$ and $u\notin A$. This is a contradiction since $u$ is present in $G-v$ and $V(G-v)\setminus S=A\dunion B$ is a consistent partition. It follows that this final case cannot happen, so $(G,X)$ is yes.
\end{proof}

It is now easy to complete our first kernelization result.

\begin{theorem}\label{theorem:stablecutsetvc:kernelization}
    \stablecutsetvc admits a polynomial kernelization to equivalent instances with at most $\Oh(|X|^3)$ vertices.
\end{theorem}

\begin{proof}
    Clearly, Reduction Rules~\ref{rrule:disconnected:cutvertex}, \ref{rrule:clique}, \ref{rrule:stableneighborhood}, and \ref{rrule:vertexcover:marking} can be exhaustively applied in polynomial time. We also know that they are safe, so the result will be the correct answer (respectively a dummy yes- or no-instance of constant size) or an equivalent instance $(G',X')$, where $|X'|\leq|X|$, to which none of the reduction rules applies. It remains to upper bound the size of $(G',X')$. Since Reduction Rule~\ref{rrule:vertexcover:marking} does not apply to $(G',X')$, we have $|V(G')\setminus X'|\leq |X'|^3\leq |X|^3$, hence $|V(G')|=\Oh(|X|^3)$.
\end{proof}

We now generalize this result to \stablecutset parameterized by the size of a twin cover. Clearly, each vertex cover of $G$ is also a twin cover.

\introduceparameterizedproblem{\stablecutsettc}{A graph $G$ and a twin cover $X$ of $G$.}{$|X|$.}{Does $G$ have a stable cutset?}

By definition, if $X$ is a twin cover of $G$ then $N_G[u]=N_G[v]$ holds for each edge $\{u,v\}$ of $G-X$. Thus each connected component of $G-X$ is a clique module of $G$. Thus, adding Reduction Rule~\ref{rrule:cliquemodule} reduces all connected components of $G-X$ (i.e., the clique modules of $G$) to at most two vertices each. It can be easily verified that Reduction Rule~\ref{rrule:vertexcover:marking} is safe also in this context and will reduce the number of isolated vertices in $G-X$ to at most $|X|^3$ many. It remains to reduce the number of connected components of size two of $G-X$. Recall that these are clique modules, i.e., if $C=\{u,v\}$ is a connected component of $G-X$ then both vertices have the same neighborhood in $X$ and no further neighbors.

\begin{rrule}\label{rrule:twincover:marking}
    If $G-X$ has more than $|X|^2$ cliques of size two, then, starting with all cliques of $G-X$ as unmarked, mark cliques of size two in $G-X$ according to the following condition:
    \begin{itemize}
        \item For each $(x_1,x_2)\in X^2$ mark a clique of $G-X$ of size two that is fully adjacent to both $x_1$ and $x_2$.
    \end{itemize}
    Pick any unmarked clique $C$ of $G-X$ of size two. If there is a $C$-simple stable cutset $S$ then return yes. Else delete $C$, i.e., return $(G-C,X)$.
\end{rrule}

\begin{lemma}
    Reduction Rule~\ref{rrule:twincover:marking} is safe, i.e., if Reduction Rules~\ref{rrule:disconnected:cutvertex}, \ref{rrule:clique}, \ref{rrule:stableneighborhood}, and \ref{rrule:cliquemodule} do not apply to $G$, and if $C$ is an unmarked clique of size two in $G-X$ then
    \begin{itemize}
        \item if there is a $C$-simple stable cutset $S$ then $(G,X)$ is yes, and
        \item else $(G,X)$ is yes if and only if $(G-C,X)$ is yes.
    \end{itemize}
\end{lemma}

\begin{proof}
    Clearly, if there is a $C$-simple stable cutset $S$ then $(G,X)$ is yes. Otherwise, by Lemma~\ref{lemma:deleteclique:nofalsenegative}, if $G$ has a stable cutset then $G-C$ has a stable cutset, i.e., if $(G,X)$ is yes then so is $(G-C,X)$. It remains to check the converse.

    Assume that $(G-C,X)$ is yes and let $S$ be a stable cutset of $G-C$ with consistent partition $V(G-C)\setminus S=A\dunion B$. If $N_G(C)\cap B=\emptyset$ then $S$ is also a stable cutset of $G$ with consistent partition $V(G)\setminus S=(A\cup C)\dunion B$. Symmetrically, if $N_G(C)\cap A=\emptyset$ then $S$ is a stable cutset of $G$ with consistent partition $V(G)\setminus S=A\dunion (B\cup C)$. It remains to consider the case that $C$ has neighbors in both $A$ and $B$. 
    
    In this case, as $N_G(C)\subseteq X$, there are $x_1\in A\cap X$ and $x_2\in B\cap X$ that are fully adjacent to $C$ (as $C$ is also a clique module of $G$). Since $C$ remained unmarked, a different connected component $C'$ of $G-X$ of size two was marked for being fully adjacent to $x_1$ and $x_2$ and its vertices are present in $G-C$. But then this full adjacency to $x_1\in A$ and $x_2\in B$ requires $C'\subseteq S$; a contradiction to $S$ being a stable set.
\end{proof}

Again, it is now easy to wrap up the kernelization result.

\begin{theorem}\label{theorem:stablecutsettc:kernelization}
    \stablecutsettc admits a polynomial kernelization to equivalent instances with at most $\Oh(|X|^3)$ vertices.
\end{theorem}

\begin{proof}
    The kernelization proceeds by exhaustive application of Reduction Rules~\ref{rrule:disconnected:cutvertex}, \ref{rrule:clique}, \ref{rrule:stableneighborhood}, \ref{rrule:cliquemodule}, \ref{rrule:vertexcover:marking}, and \ref{rrule:twincover:marking}. Clearly this can be done polynomial time (see also Observation~\ref{observation:rrule:cliquemodule:efficient}). We showed that all reduction rules are safe, which implies that they either return the correct answer (or a dummy yes- or no-instance of constant size) or an equivalent instance $(G',X')$, with $|X'|\leq|X|$, to which none of the reduction rules applies. Since Reduction Rule~\ref{rrule:vertexcover:marking} does not apply, there are at most $|X'|^3$ isolated vertices in $G'-X'$. Since Reduction Rule~\ref{rrule:twincover:marking} does not apply, there are at most $|X'|^2$ connected components of size two in $G'-X'$. We already discussed that there are no further components. Hence, $|V(G'-X')|\leq |X'|^3+2|X'|^2=\Oh(|X|^3)$.
\end{proof}

\subsection{Distance to cluster}

In this section we show a polynomial kernelization for \stablecutset parameterized by the size of a given modulator $X$ to the class of cluster graphs.

\introduceparameterizedproblem{\stablecutsetmodcluster}{A graph $G=(V,E)$ and $X\subseteq V$ s.t.\ $G-X$ is a cluster graph.}{$|X|.$}{Does $G$ have a stable cutset?}

We use Reduction Rules~\ref{rrule:disconnected:cutvertex}, \ref{rrule:clique}, \ref{rrule:stableneighborhood}, and \ref{rrule:simplicial}, as well as the following rules.

\begin{rrule}\label{rrule:cluster:simplecomponent}
    If any (clique) component $C$ of $G-X$ has a $C$-simple stable cutset then return yes.
\end{rrule}

Safeness of this rule is immediate. Since we only have to test this for connected components of $G-X$ it is also efficient (unlike testing for all cliques of a graph).
The following rule allows us to shrink connected (clique) components of $G-X$. Note that, as a special case, it deletes vertices in $G-X$ that have no neighbor in $X$, but these are already removed by Rule~\ref{rrule:simplicial} for being simplicial as all their neigbors are in their clique component in $G-X$.

\begin{rrule}\label{rrule:shrinkcliques}
Let $C$ be a (clique) component of $G-X$ and let $v\in C$. If each neighbor $x\in X$ of $v$ is adjacent to at least two other vertices of $C$ then delete $v$, i.e., return the instance $(G-v,X)$.
\end{rrule}

\begin{lemma}
    Reduction Rule~\ref{rrule:shrinkcliques} is safe.
\end{lemma}

\begin{proof}
By Lemma~\ref{lemma:deletevertex:nofalsenegative}, if $(G,X)$ is yes then so is $(G-v,X)$.

Assume now that $(G-v,X)$ is yes, let $S$ be a stable cutset of $G-v$, and let $A\dunion B$ be a consistent partition of $(G-v)-S$. By requirement of the reduction rule, $C\setminus\{v\}$ is a clique of size at least two in $G-v$ and, hence, at most one of its vertices can be in $S$. Clearly, the nonempty rest of $C\setminus\{v\}$ must be fully contained in $A\cup S$ or in $B\cup S$; w.l.o.g., it is a subset of $A\cup S$. Let $x\in X$ be any neighbor of $v$ (in $G$). By requirement of the reduction rule, there are two other vertices $u,w\in C\setminus\{v\}$ that are both adjacent to $x$ (in both $G$ and $G-v$). As $u$ and $w$ are adjacent, w.l.o.g., we have $u\notin S$ and hence $u\in A$. It follows directly that $x\in A\cup S$ and, by extension, that all neighbors of $v$ (in $G$) are contained in $A\cup S$. Therefore, $S$ is also a stable cutset of $G$ as witnessed by consistent partition $(A\cup\{v\})\dunion B$ for $G-S$.
\end{proof}

\begin{observation}
    If Reduction Rule~\ref{rrule:shrinkcliques} does not apply to $(G,X)$ then each (clique) component $C$ of $G-X$ has size at most $2|X|$: Greedily pick a set $C'\subseteq C$ by adding any vertex $v\in C\setminus C'$ that has a neighbor $x\in X$ that is not yet shared by two or more vertices in $C'$. Clearly, each $x\in X$ can play this role at most twice, which implies that we terminate with some $C'\subseteq C$ of size at most $2|X|$. If $|C|>2|X|$ then the rule can be applied to any vertex $v\in C\setminus C'\neq\emptyset$.
\end{observation}

The following more involved marking-based reduction rule allows us to reduce the number of connected (clique) components of $(G-X)$.

\begin{rrule}\label{rrule:markingcliques}
If $G-X$ has more than $|X|^3+3|X|^4$ (clique) components, then, starting with all vertices as unmarked, mark vertices of $G-X$ according to the following conditions:
\begin{itemize}
    \item For each $(x_1,x_2,x_3)\in X^3$ mark any vertex $p$ that is adjacent to $x_1$, $x_2$, $x_3$.
    \item For each $(x_1,x_2,x_3,x_4)\in X^4$ mark any two vertices $p$ and $q$ that are adjacent, and with $p$ adjacent to $x_1$ and $x_2$ and $q$ adjacent to $x_3$ and $x_4$.
    \item For each $(x_1,x_2,x_3,x_4)\in X^4$ mark any three vertices $p$, $q$, and $r$ that are pairwise adjacent, and with $p$ adjacent to $x_1$ and $x_2$, $q$ adjacent to $x_3$, and $r$ adjacent to $x_4$.
    \item For each $(x_1,x_2,x_3,x_4)\in X^4$ mark any four vertices $p$, $q$, $r$, and $s$ that are pairwise adjacent, and with $p$ adjacent to $x_1$, $q$ adjacent to $x_2$, $r$ adjacent to $x_3$, and $s$ adjacent to $x_4$.
\end{itemize}
Delete any (clique) component $C$ of $G-X$ whose vertices are all unmarked, i.e., return the instance $(G-C,X)$.
\end{rrule}

\begin{lemma}
    Reduction Rule~\ref{rrule:markingcliques} is safe.
\end{lemma}

\begin{proof}
Let $C$ be the (clique) component of $G-X$ that is being deleted. As Reduction Rule~\ref{rrule:cluster:simplecomponent} did not apply, by Lemma~\ref{lemma:deleteclique:nofalsenegative}, if $G$ has a stable cutset then so does $G-C$, i.e., if $(G,X)$ is yes then so is $(G-C,X)$.

Now assume that $(G-C,X)$ is yes, let $S$ be a stable cutset of $G-C$, and let $A\dunion B$ be a consistent partition of $(G-C)-S$. We show via case distinction that we can always find a stable cutset $S'\supseteq S$ of $G$ with a consistent partition $A'\cup B'$ such that $A'\supseteq A\neq\emptyset$ and $B'\supseteq B\neq\emptyset$. Intuitively, we show that the vertices of $C$ can be placed consistently with solution $S$ and partition $A\cup B$.

Such a solution $S'$ can be trivially found in two (symmetric) cases: If $N(C)\subseteq A\cup S$ then let $S'=S$, $A'=A\cup C$, and $B'=B$. Symmetrically, if $N(C)\subseteq S\cup B$ then let $S'=S$, $A'=A$, and $B'=B\cup C$. 

In the remaining cases, $C$ contains at least one vertex adjacent to $A$ and at least one vertex adjacent to $B$ (possibly the same vertex). It will be convenient to partition $C$ according to neighborhood in $A$, $S$, and $B$, namely $C=C_{ASB}\dunion C_{AS}\dunion C_{SB}\dunion C_{AB}\dunion C_A\dunion C_B\dunion C_S$. E.g., each vertex in $C_{AS}$ has at least one neighbor in $A$, at least one neighbor in $S$, and no neighbor in $B$. Note that $N(C)\subseteq X$ as $C$ is a connected component of $G-X$, and that each vertex in $C$ has a neighbor in $X$ as Reduction Rule~\ref{rrule:simplicial} does not apply.

If $C_{ASB}\neq\emptyset$, let $p\in C_{ASB}$ and let $x_1\in A\cap X$, $x_2\in S\cap X$, and $x_3\in B\cap X$ be corresponding neighbors of $p$. Since $p$ is unmarked while being eligible to be marked for $(x_1,x_2,x_3)\in X^3$ there must be a vertex $p'$ that was marked for this triplet. As all vertices of $C$ were unmarked, $p'\notin C$ and, hence, $p'$ is a vertex of $G-C$. This, however, yields a contradiction since due to having neighbors in $A$, $S$, and $B$, vertex $p'$ cannot be in any of the three sets itself. It follows that $C_{ASB}=\emptyset$.

If $|C_{AB}|\geq 2$, let $p\in C_{AB}$ with corresponding neighbors $x_1\in A\cap X$ and $x_2\in B\cap X$ and let $q\in C_{AB}\setminus\{p\}$ with corresponding neighbors $x_3\in A\cap X$ and $x_4\in B\cap X$. Then $p$ and $q$ are eligible for being marked for the tuple $(x_1,x_2,x_3,x_4)\in X^4$ but remained unmarked (along with the rest of $C$). It follows that there are vertices $p'$ and $q'$ that were marked for $(x_1,x_2,x_3,x_4)$ and that they are present in $G-C$. Again, this yields a contradiction: Being adjacent, at most one of $p'$ and $q'$ can be in $S$, but then the other is not in $S$ but also adjacent to $A$ and $B$. It follows that $|C_{AB}|\leq 1$.

If $|C_{AB}|=1$, let $p\in C_{AB}$ with corresponding neighbors $x_1\in A\cap X$ and $x_2\in B\cap X$ (and no neighbor in $S$); we note that $C_{AB}=\{p\}$. (1) If no vertex in $C\setminus\{p\}$ has a neighbor in $B$ then we find the stable cutset $S'=S\cup\{p\}$ with consistent partition $A'=A\cup(C\setminus\{p\})$ and $B'=B$; note that in this case for each $q\in C\setminus\{p\}$ we have $N(q)\subseteq C\cup A\cup S=A'\cup S'$. (2) Symmetrically, if no vertex in $C\setminus\{p\}$ has a neighbor in $A$ then we get solution $S'=S\cup\{p\}$ with consistent partition $A'=A$ and $B'= B\cup(C\setminus\{p\})$. (3) If neither subcase (1) nor (2) applies, then $C\setminus\{p\}$ contains a vertex $q$ with neighbor $x_3\in A$ and a vertex $r$ with neighbor $x_4\in B$; clearly, $x_3\in A\cap X$ and $x_4\in B\cap X$ as $N(C)\subseteq X$. Note that $q\neq r$ as otherwise there would be a second vertex adjacent to $A$ and $B$ while $C_{AB}=\{p\}$. Since $p,q,r$ are eligible to be marked for $(x_1,x_2,x_3,x_4)\in X^4$ but are unmarked (along with the rest of $C$), some other vertices $p'$, $q'$, and $r'$ were marked for having this adjacency to $x_1$ through $x_4$, and they are present in $G-C$. This yields a contradiction: Out of $p'$, $q'$, and $r'$ at most one vertex can be in $S$, always leaving an adjacent pair of vertices that are adjacent to vertices in $A$ and $B$. Thus, when $|C_{AB}|=1$ we find a solution for $(G,X)$ as the remaining case leads to a contradiction.

It remains to consider the case that $C_{ABS}=C_{AB}=\emptyset$, i.e., that $C=C_{AS}\cup C_{SB}\cup C_A\cup C_B\cup C_S$. We first consider the case that $C_{AS}\neq\emptyset$, and $C_{SB}\neq\emptyset$ is symmetric. Afterwards, only $C=C_A\cup C_B\cup C_S$ remains.

If $C_{AS}\neq\emptyset$ let $p\in C_{AS}$ with corresponding neighbors $x_1\in A\cap X$ and $x_2\in S\cap X$. (1) If also $C_{SB}\neq \emptyset$ then let $q\in C_{SB}$ with corresponding neighbors $x_3\in S\cap X$ and $x_4\in B\cap X$. As $p$ and $q$ (and the rest of $C$) are unmarked, we have marked some other adjacent vertices $p'$ and $q'$ with the same adjacency. These, however, are present in $G-C$ and yield a contradiction: Neither $p'$ nor $q'$ can be in $S$, but then they form a path between $A$ and $B$. Thus, if $C_{AS}\neq\emptyset$ then $C_{SB}=\emptyset$. (2) If $|C_B|\geq 2$ then let $q\in C_B$ with corresponding neighbor $x_3\in B\cap X$ and let $r\in C_B\setminus\{q\}$ with corresponding neighbor $x_4\in B\cap X$. Again, there are marked pairwise adjacent vertices $p'$, $q'$, and $r'$ with same adjacency that are present in $G-C$ and we get a contradiction: We have $p'\notin S$, by adjacency to $x_2\in S$, and only one of the adjacent vertices $q'$ and $r'$ can be in $S$, leaving a path from $A$ via $p'$ and via $q'$ or $r'$ to $B$. Thus, if $C_{AS}\neq\emptyset$ then $|C_B|\leq 1$. (3) In the remaining case we thus have $C=C_{AS}\cup C_A\cup C_B\cup C_S$ where $|C_B|\leq 1$ but then $|C_B|=1$ as we already dealt with the case that $C$ has no neighbors in $B$. Then letting $C_B=\{q\}$, we get the stable cutset $S'=S\cup\{q\}$ with consistent partition $A'=A\cup(C\setminus\{q\})$ and $B'=B$.

Symmetrically, if $C_{SB}\neq\emptyset$ then we can rule out $C_{AS}\neq\emptyset$ and $|C_A|\geq 2$. For the remaining case of $C=C_{SB}\cup C_A\cup C_B\cup C_S$ with $|C_A|=1$. Letting $C_A=\{q\}$ we get the stable cutset $S'=S\cup\{q\}$ with consistent partition $A'=A$ and $B'=B\cup(C\setminus\{q\})$.

We arrive at the case that $C=C_A\cup C_B\cup C_S$. If $|C_A|\leq 1$ then we find a stable cutset $S'=S\cup C_A$ with consistent partition $A'=A$ and $B'=B\cup C_B\cup C_S$. Symmetrically, if $|C_B|\leq 1$ then there is a stable cutset $S'=S\cup C_B$ with consistent partition $A'=A\cup C_A\cup C_S$ and $B'=B$. Otherwise, if $|C_A|\geq 2$ and $|C_B|\geq 2$ we can identify a contradiction: Let $p\in C_A$ with neighbor $x_1\in A\cap X$, let $q\in C_A\setminus\{p\}$ with neighbor $x_2\in A\cap X$, let $r\in C_B$ with neighbor $x_3\in B\cap X$, and let $s\in C_B\setminus\{r\}$ with neighbor $x_4\in B\cap X$. As these are unmarked but eligible for $(x_1,x_2,x_3,x_4)$ with this adjacency, we must have marked instead four pairwise adjacent vertices $p'$, $q'$, $r'$, and $s'$ with the same adjacency and these are present in $G-C$. Clearly, only one of the four vertices can be in $S$ and the remaining ones always yield a path from $A$ to $B$ avoiding $S$; a contradiction.

Each case was either showed to be impossible or we identified a stable cutset for $G$, showing that $(G,X)$ is yes. This completes the proof.
\end{proof}

We can now complete the kernelization for \stablecutsetmodcluster.

\begin{theorem}
\stablecutsetmodcluster admits a polynomial kernelization to equivalent instances with $\Oh(|X|^5)$ vertices.
\end{theorem}

\begin{proof}
Given an instance $(G,X)$ we exhaustively apply Reduction Rules~\ref{rrule:disconnected:cutvertex}, \ref{rrule:clique}, \ref{rrule:stableneighborhood}, \ref{rrule:simplicial}, \ref{rrule:cluster:simplecomponent}, \ref{rrule:shrinkcliques}, and \ref{rrule:markingcliques}. We already know that the reduction rules are safe, and it is straightforward to implement the kernelization to run in polynomial time. When the kernelization terminates we receive the correct yes or no answer (or a dummy yes- or no-instance of constant size) or an equivalent instance $(G',X')$ with $X'\subseteq X$. It remains to bound the number of vertices of $G'$ in the latter case.

Since Rule~\ref{rrule:markingcliques} does not apply, the graph $G'-X'$ has at most $|X|^3+3|X|^4=\Oh(|X|^4)$ (clique) components. As Rule~\ref{rrule:shrinkcliques} does not apply each such clique has at most $2|X'|$ vertices. Since $|X'|\leq|X|$, we get a total of $\Oh(|X'|+|X'|\cdot |X'|^4)=\Oh(|X|^5)$ vertices.
\end{proof}

As a free corollary we get a vertex-linear kernelization for \stablecutset parameterized by the size of a modulator to a single clique.

\introduceparameterizedproblem{\stablecutsetmodclique}{A graph $G=(V,E)$ and $X\subseteq V$ such that $G-X$ is a clique.}{$|X|$.}{Does $G$ have a stable cutset?}

\begin{corollary}
    \stablecutsetmodclique has a polynomial kernelization to equivalent instances with $\Oh(|X|)$ vertices.
\end{corollary}

\begin{proof}
    We apply the kernelization for \stablecutsetmodcluster and observe that it never introduces additional components to $G-X$. Thus, the single clique in $G-X$ is reduced to at most $2|X|$ vertices, for a total of at most $3|X|$ vertices.
\end{proof}

\subsection{Distance to co-cluster}

In this section we show a polynomial kernelization for \stablecutset parameterized by the size of a modulator to the class of co-cluster graphs (the complements of disjoint unions of cliques).

\introduceparameterizedproblem{\stablecutsetmodcocluster}{A graph $G=(V,E)$ and $X\subseteq V$ s.t.\ $G-X$ is a co-cluster.}{$|X|$.}{Does $G$ have a stable cutset?}

We use Reduction Rules~\ref{rrule:disconnected:cutvertex}, \ref{rrule:clique}, and \ref{rrule:stableneighborhood} as well as the following rules. For each instance $(G,X)$, we have that $\overline{G-X}$ is a disjoint union of cliques, while $G-X$ is a join of stable sets. If $G-X$ is a (single) stable set, $X$ is also a vertex cover of $G$ and we can apply the kernelization for \stablecutsetvc. Else, the following reduction rule can be used to reduce the size of the stable sets of $G-X$.

\begin{rrule}\label{rrule:cocluster:shrinkindset}
    If $G-X$ is the join of at least two stable sets, 
    and if one such stable set $I$ has more than $|X|^2$ vertices then, starting with all vertices of $I$ as unmarked, mark vertices of $I$ according to the following condition:
    \begin{itemize}
        \item For each $(x_1,x_2)\in X^2$ mark any vertex $p\in I$ that is adjacent to $x_1$ and $x_2$.
    \end{itemize}
    Delete any unmarked vertex $v$ of $I$, i.e., return $(G-v,X)$.
\end{rrule}

\begin{lemma}
    Reduction Rule~\ref{rrule:cocluster:shrinkindset} is safe, i.e., if Reduction Rules~\ref{rrule:disconnected:cutvertex}, \ref{rrule:clique}, and \ref{rrule:stableneighborhood} do not apply, if $G-X$ is the join of at least two stable sets, and if $v$ remains unmarked during an application of Reduction Rule~\ref{rrule:cocluster:shrinkindset}, then $(G,X)$ is yes if and only if $(G-v,X)$ is yes.
\end{lemma}

\begin{proof}
    By Lemma~\ref{lemma:deletevertex:nofalsenegative} if $G$ has a stable cutset then so does $G-v$, i.e., if $(G,X)$ is yes then so is $(G-v,X)$. It remains to show the converse.

    Assume that $(G-v,X)$ is yes and let $S$ be a stable cutset of $G-v$ with consistent partition $V(G-v)\setminus S=A\dunion B$. If $N_G(v)\cap A=\emptyset$ then $S$ is a stable cutset of $G$ with consistent partition $V(G)\setminus S=A\dunion (B\cup\{v\})$. Symmetrically, if $N_G(v)\cap B=\emptyset$ then $S$ is a stable cutset of $G$ with consistent partition $V(G)\setminus S=(A\cup\{v\})\dunion B$. If $N_G(v)\cap S=\emptyset$ then $S\cup\{v\}$ is a stable cutset of $G$ with consistent partition $V(G)\setminus(S\cup\{v\})=A\dunion B$. It remains to consider the case that $v$ has neighbors in all three sets $S$, $A$, and $B$.

    Recall the assumption that $G-X$ is the join of at least two stable sets and let $J=V(G-X)\setminus I\neq\emptyset$ be the union of the stable sets other than $I$. Note that $I$ and $J$ are fully adjacent, and that $J\subseteq N_G(v)\subseteq X\cup J$. We make a case distinction into case (i) that $(N_G(v)\cap S)\cap X\neq\emptyset$ and case (ii) that $(N_G(v)\cap S)\cap X=\emptyset$, i.e., $N_G(v)\cap S\subseteq J$.

    We first consider case (i) that $(N_G(v)\cap S)\cap X\neq\emptyset$. Let $x_1\in (N_G(v)\cap S)\cap X$. Since $v$ was not marked for $(x_1,x_1)\in X^2$, there is a vertex $p\in I\setminus\{v\}$ that was marked for this pair (and is adjacent to $x_1$). As $x_1\in S$ we must have $p\notin S$. W.l.o.g., assume that $p\in A$. Since $p\in I$, it is fully adjacent to $J$, which implies that $J\cap B=\emptyset$.

    Since $v$ has a neighbor in $B$ and $N_G(v)\subseteq X\cup J$, it follows that $v$ has a neighbor $x_2\in X\cap B$. Since $v$ is unmarked but eligible for $(x_1,x_2)$, some other vertex $q\in I\setminus\{v\}$ was marked for the pair $(x_1,x_2)\in X^2$ and is adjacent to both $x_1$ and $x_2$. As $x_1\in S$ and $x_2\in B$, we have that $q\in B$ (and hence that $q\neq p$). Since $q\in I$, it is fully adjacent to $J$, which implies that $J\cap A=\emptyset$. But then $J\subseteq S$, i.e., it must be a single stable set (and $G-X$ is the join of the stable sets $I$ and $J$).

    Since $v$ has a neighbor in $A$ and $N_G(v)\subseteq X\cup J$, it now follows that $v$ has a neighbor $x_3\in X\cap A$. Since $v$ is unmarked but eligible for $(x_2,x_3)$, some other vertex $r\in I\setminus\{v\}$ was marked for the pair $(x_2,x_3)\in X^2$ and is adjacent to both $x_2$ and $x_3$. As $x_2\in B$ and $x_3\in A$, we have that $r\in S$. This however is a contradiction since $r$ is fully adjacent to $J\subseteq S$ and $J\neq\emptyset$. This shows that case (i) is impossible.

    We now consider case (ii) that $(N_G(v)\cap S)\cap X=\emptyset$, which implies that $N_G(v)\cap S\subseteq J$ (as $N_G(v)\subseteq X\cup J$). Since $v$ does have at least one neighbor in $S$, hence in $S\cap J$, it follows that $I\setminus\{v\}\subseteq A\cup B$ as $I$ is fully adjacent to $J$. Assume for a contradiction that $J\subseteq S$. Then $v$ must have a neighbor $x_1\in X\cap A$ and a neighbor $x_2\in X\cap B$ (as it has at least one neighbor each in $A$ and $B$, and $N_G(v)\subseteq X\cup J$). Since $v$ was not marked for $(x_1,x_2)\in X$, some other vertex $p\in I\setminus\{v\}$ was marked for this pair and is adjacent to $x_1$ and $x_2$. As $x_1\in A$ and $x_2\in B$, we must have $p\in S$. This is a contradiction as $p\in I$ is fully adjacent to $J\subseteq S$ and $J\neq\emptyset$. It remains to consider the alternative that $J\nsubseteq S$.

    This implies that $J\cap A\neq\emptyset$ and/or $J\cap B\neq\emptyset$. Assume w.l.o.g.\ that $J\cap A\neq\emptyset$. Then $J$ contains vertices from $A$ and $S$ and, hence, $I\setminus\{v\}\subseteq A$ as $I$ and $J$ are fully adjacent. Similarly, with $J$ being fully adjacent to $I\setminus\{v\}\subseteq A$, we have $J\cap B=\emptyset$, hence $J\subseteq A\cup S$. Since $v$ has a neighbor in $B$ and $N_G(v)\subseteq X\cup J$, it follows that $v$ has a neighbor $x_3\in X\cap B$. Since $v$ was not marked for $(x_3,x_3)\in X^2$, some other vertex $q\in I\setminus\{v\}$ was marked for this pair and is adjacent to $x_3$. This, however, is a contradiction since $q\in I\setminus\{v\}\subseteq A$ has a neighbor $x_3\in B$. Thus, case (ii) is equally impossible.

    This shows that only the three cases of the first paragraph are possible, and in each of them a stable cutset for $G$ was found. Thus $(G,X)$ is yes if $(G-v,X)$ is yes.
\end{proof}

We use a second reduction rule to reduce the \emph{number} of stable sets whose join is $G-X$. This rule is an adaptation of Reduction Rule~\ref{rrule:shrinkcliques} that we used for \stablecutsetmodcluster. Note that, as a special case, it removes vertices in $G-X$ that have no neighbor in $X$, i.e., with $N_G(v)\cap X=\emptyset$.

\begin{rrule}\label{rrule:cocluster:reducestablesets}
    If $G-X$ is the join of at least four stable sets and if there is a vertex $v\in V(G-X)$ such that each neighbor $x\in N_G(v)\cap X$ is adjacent to at least three different stable sets of $G-X$ (including the one of $v$) then delete $v$, i.e., return $(G-v,X)$.
\end{rrule}

\begin{lemma}\label{lemma:cocluster:reducestablesets}
    Reduction Rule~\ref{rrule:cocluster:reducestablesets} is safe, i.e., if Reduction Rules~\ref{rrule:disconnected:cutvertex}, \ref{rrule:clique}, and \ref{rrule:stableneighborhood} do not apply and if $v$ fulfills the condition of the rule then $(G,X)$ is yes if and only if $(G-v,X)$ is yes.
\end{lemma}

\begin{proof}
    By Lemma~\ref{lemma:deletevertex:nofalsenegative} if $G$ has a stable cutset then so does $G-v$, hence if $(G,X)$ is yes then so is $(G-v,X)$. It remains to verify the converse.

    Now, assume that $(G-v,X)$ is yes and let $S$ be a stable cutset of $G-v$ with consistent partition $V(G-v)\setminus S=A\dunion B$. Let $Y=V(G)\setminus X$ and observe that the graph $G[Y\setminus\{v\}]=(G-v)-X$ is a complete $k$-partite graph with $k\geq3$ as $G-X$ is complete $\ell$-partite with $\ell\geq 4$. As $S$ is a stable set, it contains vertices from at most one partite set of $G[Y\setminus\{v\}]$ and hence $G[Y\setminus\{v\}]-S$ remains connected. Thus, it must hold that $(Y\setminus\{v\})\setminus S\subseteq A$ or $(Y\setminus\{v\})\setminus S\subseteq B$; w.l.o.g., assume that $(Y\setminus\{v\})\setminus S\subseteq A$. It also follows that each vertex in $N_G(v)\cap X$ is adjacent to a vertex in $(Y\setminus\{v\})\setminus S$ because $(Y\setminus\{v\})\setminus S$ completely contains all but at most two partite sets of $G-X$. Thus, such vertices cannot be in $B$ and we have $N_G(v)\cap X\subseteq S\cup A$.
    At the same time, $N_G(v)\setminus X=N_G(v)\cap Y\subseteq Y\setminus\{v\}\subseteq S\cup A$. Thus, $N_G(v)\subseteq S\cup A$, which implies that $S$ is also a stable cutset for $G$ with consistent partition $V(G)\setminus S=(A\cup\{v\})\dunion B$.
\end{proof}

We are now ready to complete the kernelization.

\begin{theorem}
    \stablecutsetmodcocluster admits a polynomial kernelization to equivalent instances with $\Oh(|X|^3)$ vertices.
\end{theorem}

\begin{proof}
    We exhaustively use Reduction Rules~\ref{rrule:disconnected:cutvertex}, \ref{rrule:clique}, \ref{rrule:stableneighborhood}, \ref{rrule:cocluster:shrinkindset}, and  \ref{rrule:cocluster:reducestablesets}. Additionally, if at any point $G-X$ is a stable set then $X$ is also a vertex cover of $G$ and we proceed as in the kernelization for \stablecutsetvc. (More formally, we can append the reduction rules for \stablecutsetvc and set them to apply only when $G-X$ is a stable set.) Clearly, all reduction rules can be applied exhaustively in polynomial time. As they are all safe, we get the correct yes or no answer (or a dummy yes- or no-instance of constant size) or an equivalent instance $(G',X')$ with $X'\subseteq X$. It remains to bound the size of $G'$ in the latter case.

    If $G'-X'$ is a stable set then we already know that $G'$ has $\Oh(|X'|^3)=\Oh(|X|^3)$ vertices as the reduction rules for \stablecutsetvc do not apply (and not for $G'-X'$ not being a stable set). Else, $G'-X'$ is the join of at least two stable sets, and since Reduction Rule~\ref{rrule:cocluster:shrinkindset} does not apply, we know that each (maximal) stable set of $G'-X'$ has at most $|X'|^2$ vertices. If $G'-X'$ is the join of at most three stable sets then this already allows to upper bound the number of vertices by $|X'|+3\cdot |X'|^2=\Oh(|X'|^2)=\Oh(|X|^2)$. It remains to bound the number of vertices in $G'$ when $G'-X'$ is the join of at least four stable sets.

    In this case, since Reduction Rule~\ref{rrule:cocluster:reducestablesets} does not apply, we know that for each vertex $v\in Y:=V(G'-X')$ there is at least one vertex $x\in N_{G'}(v)\cap X'$ that is adjacent to at most two different stable sets of $G'-X'$. Since each stable set has at most $|X'|^2$ vertices, we get $|N_{G'}(x)\cap Y|\leq 2|X'|^2$. Overall, since such an adjacent vertex $x\in X$ exists for each $v\in Y$, we get that $|Y|\leq |X|\cdot 2|X'|^2=\Oh(|X'|^3)$. Together with $X'$, we get the same upper bound of $\Oh(|X|^3)$ vertices.
\end{proof}

\section{Lower bounds for kernelizations for Stable Cutset}\label{section:lowerbounds}

Our main result in this section is a lower bound for kernelization for \stablecutset parameterized by the deletion distance to a linear forest. As a free corollary, our construction also provides a lower bound for parameterization by the size of a dominating set. We also show simple lower bounds relative to parameters treewidth and the size of the stable cutset.

\subsection{Parameterization by distance to a linear forest}

In this section, we show that \stablecutsetmodlinforest is hard under polynomial parameter transformations for the class \mktwo introduced by Hermelin et al.~\cite{HermelinKSWW15}. This implies that unless \containment (and the polynomial hierarchy collapses) it has no polynomial kernelization~\cite{HermelinKSWW15}.

\introduceparameterizedproblem{\stablecutsetmodlinforest}{A graph $G=(V,E)$ and $X\subseteq V$ s.t.\ $G-X$ is a linear forest.}{$|X|$.}{Does $G$ have a stable cutset?}

The class \mktwo can be defined via its complete problem \hittingsetn, i.e., it consists of all parameterized problems that have a PPT to \hittingsetn.

\introduceparameterizedproblem{\hittingsetn}{A ground set $U$, a family $\F$ of subsets of $U$, and $k\in\N$.}{$n=|U|$.}{Is there a hitting set $S\subseteq U$ for $\F$ of size at most $k$?}

We establish \mktwo-hardness of \stablecutsetmodlinforest via the intermediate problems \multicoloredhittingsetn and \setsplittingn (see further below), both of which we show to be complete for \mktwo under PPTs, which may be of independent interest (though not at all surprising).

\introduceparameterizedproblem{\multicoloredhittingsetn}{A ground set $U$, a family $\F$ of subsets of $U$, $k\in\N$, and\\
& a coloring function $\phi\colon U\to[k]$.}{$n=|U|$.}{Is there a hitting set $S\subseteq U$ for $\F$ that contains exactly \\
&one element of each color $i\in[k]$?}

\begin{lemma}\label{lemma:multicoloredhittingsetn:mktwocomplete}
    \multicoloredhittingsetn is complete for \mktwo. 
\end{lemma}

\begin{proof}
    We give simple PPTs between \hittingsetn and \multicoloredhittingsetn, which establish that the latter is complete for \mktwo under PPTs. 
    
    For the first reduction, let $(U,\F,k)$ be an instance of \hittingsetn with $n=|U|$ and $m=|\F|$, asking whether there is a subset $S\subseteq U$ of size at most $k$ that has a nonempty intersection with each set $F\in\F$. 
    If $k\geq n$ then $(U,\F,k)$ is trivially yes so return a yes instance of \multicoloredhittingsetn of constant size. Else, construct an instance $(U',\F',k,\phi)$ of \multicoloredhittingsetn as follows:
    \begin{itemize}
        \item Let $U'=\{(u,i) \mid u\in U \wedge i\in[k]\}$ and let $\phi\colon U'\to[k]\colon (u,i)\mapsto i$. In other words, for each element $u\in U$ create $k$ new elements $(u,1),\ldots,(u,k)$, one of each color $i\in[k]$.
        \item For each $F\in\F$ let $\F'$ contain the set $F'=\{(u,i)\mid u\in F\}$, i.e., the sets $F'$ of $\F'$ simply contain the elements of the corresponding set $F$ in all $k$ colors.
    \end{itemize}
    Clearly, this can be done in polynomial time and the parameter value of instance $(U',\F',k,\phi)$ is $n'=|U'|\leq k\cdot |U|\leq n^2$ which is polynomially bounded in the input parameter $n$. It remains to (briefly) verify correctness.

    Assume first that $(U,\F,k)$ is yes and let $S\subseteq U$ be of size at most $k$ have a nonempty intersection with each set $F\in \F$. Let $S=\{u_1,\ldots,u_{k'}\}$ with $k'\leq k$ and let $S'=\{(u_1,1),\ldots,(u_{k'},k'),(u_{k'},k'+1),\ldots,(u_{k'},k)\}$. Clearly, $S'$ has a nonempty intersection with each set $F'\in\F'$ and it contains exactly one element of each color. Thus $(U',\F',k,\phi)$ is yes.

    Now assume that $(U',\F',k,\phi)$ is yes and let $S'\subseteq U'$ be a multicolored hitting set for $\F'$. Let $S'=\{(u_1,1),\ldots,(u_k,k)\}$ and let $S=\{u\mid (u,i)\in S'\}\subseteq U$. Clearly $S$ is of size at most $k$ and has a nonempty intersection with each set $F\in\F$. Thus $(U,\F,k)$ is yes.
    This completes the PPT from \hittingsetn to \multicoloredhittingsetn.

    For the reverse reduction, let $(U,\F,k,\phi)$ be an instance of \multicoloredhittingsetn. We construct an instance $(U,\F',k)$ of \hittingsetn as follows:
    \begin{itemize}
        \item Start from $\F'=\F$. For each color $i\in[k]$ add to $\F'$ a set $F_i$ containing all elements of $U$ of color $i$, formally $F_i=\{u\mid \phi(u)=i\}$.
    \end{itemize}
    Clearly, this can be done in polynomial time and the parameter value of the output instance is the same as for the input. Correctness is straightforward: A multicolored hitting set for the input also hits the additional sets $F_i$ in the output instance and has size at most $k$. Conversely, a hitting set of size at most $k$ for the output instance hits all the sets $F_i$ and thereby must contain exactly one element per color (while hitting all sets in $\F\subseteq\F'$). This completes the second PPT and the proof of \mktwo-completeness of \multicoloredhittingsetn.
\end{proof}

We now turn to proving \mktwo-completeness of \setsplittingn under PPTs. For proving \mktwo-hardness we adapt a reduction due to Cygan et al.~\cite{CyganDLMNOPSW16}.

\introduceparameterizedproblem{\setsplittingn}{A ground set $U$ and a family $\F$ of subsets of $U$.}{$n=|U|$.}{Is there $S\subseteq U$ with $F\nsubseteq S$ and $F\nsubseteq U\setminus S$ for all $F\in\F$?}

\begin{lemma}
    \setsplittingn is complete for \mktwo under PPTs.
\end{lemma}

\begin{proof}
    To show \mktwo-hardness of \setsplittingn, we adapt a reduction due to Cygan et al.~\cite{CyganDLMNOPSW16} (from bounded arity \hittingset to bounded arity \setsplitting) to get a polynomial parameter transformation from \multicoloredhittingsetn to \setsplittingn.

    Let $(U,\F,k,\phi)$ be an instance of \multicoloredhittingsetn, with $\phi\colon U\to[k]$, asking whether there is a hitting set $S\subseteq U$ for $\F$ that contains exactly one element of each color. If there are any empty color classes then the instance is trivially no and we return a dummy no-instance of \setsplittingn.
    Else, we construct an instance $(U',\F')$ of \setsplittingn as follows:
    \begin{itemize}
        \item Let $U'=U\cup\{r,b\}$ and let $U_i=\{u\in U \mid \phi(u)=i\}$.
        \item Start with $\F'$ containing just the set $\{r,b\}$.
        \item For each $i\in[k]$, add to $\F'$ all sets $\{u,u',r\}$ with $\{u,u'\}\in\binom{U_i}{2}$.
        \item For each $F\in\F$, add to $\F'$ the set $F\cup\{b\}$.
    \end{itemize}
    Clearly this can be done in polynomial time and the parameter value of the output instance, namely $n'=|U'|=n+2$, is polynomially bounded in the input parameter value $n$. It remains to verify correctness, i.e., that $(U,\F,k,\phi)$ is yes for \multicoloredhittingsetn if and only if $(U',\F')$ is yes for \setsplittingn.

    Assume first that $(U,\F,k,\phi)$ is yes and let $S\subseteq U$ be a hitting set of $\F$ containing exactly one element $u_i$ of each color $i\in[k]$. We claim that $S'=S\cup\{r\}=\{u_1,\ldots,u_k,r\}$ splits all sets $F'\in\F'$, i.e., $F'\nsubseteq S'$ and $F'\nsubseteq U'\setminus S'$ for all sets $F'\in\F'$. We check this for each type of set:
    \begin{itemize}
        \item Clearly $\{r,b\}$ is split as $r\in S'$ but $b\notin S'$.
        \item Each set $\{u,u',r\}$ with $\{u,u'\}\in\binom{U_i}{2}$ is split as $r\in S'$ but at least one of $u$ and $u'$ is not in $S'$ as it contains exactly one element from each set $U_i$ (the elements of color $i$ in $U$).
        \item For each set $F'\in\F'$ with $F'=F\cup\{b\}$ for some $F\in\F$ we know that $S$, and thereby $S'$, contain some element $u\in U$ that is also in $F\subseteq F'$. At the same time, $b\notin S'$ so the set is split.
    \end{itemize}
    Thus, $(U',\F')$ is yes for \setsplittingn.

    Now assume that $(U',\F')$ is yes and let $S'\subseteq U'$ such that all sets $F'\in\F'$ are split, i.e., $F'\nsubseteq S'$ and $F'\nsubseteq U'\setminus S'$ for all sets $F'\in\F'$. Assume w.l.o.g.\ that $r\in S'$ since otherwise its complement $U'\setminus S'$ could be used instead. We show that $S=S'\setminus\{r\}$ is a hitting set of size at most $k$ that contains at most one element of each color:
    \begin{itemize}
        \item We have $b\notin S$ since $\{r,b\}$ is split by $S'$ while $r\in S'$, hence $b\notin S'\supseteq S$. Thus $S\subseteq U$.
        \item If $S$ were to contain two elements of the same color, say $u$ and $u'$ with $\phi(u)=\phi(u')=i$, then $\{u,u'\}\in \binom{U_i}{2}$ and $S'\supseteq S\supseteq\{u,u'\}$ would not split the corresponding set $\{u,u',r\}$ as also $r\in S'$.
        \item For each set $F\in\F$ we know that $S'$ splits $F\cup\{b\}$ while $b\notin S'$. Thus $S'$ contains some element $u\in F$ while clearly $u\neq r$ hence $u\in S=S'\setminus\{r\}$.
    \end{itemize}
    Thus, $S$ is a hitting set for $\F$ with at most one element per color. This directly implies that there is a hitting set for $\F$ containing exactly one element per color and that $(U,\F,k,\phi)$ is yes for \multicoloredhittingsetn. This completes the proof part for \mktwo-hardness of \setsplittingn.

    To establish membership of \setsplittingn in \mktwo it suffices to give a PPT to any problem in \mktwo. E.g., a straightforward reduction from \setsplittingn to \cnfsatn (e.g., as in~\cite[Observation 3.6]{CyganDLMNOPSW16}) suffices.
\end{proof}

We are now ready to prove \mktwo-hardness of \stablecutsetmodlinforest by giving a PPT from \setsplittingn.

\begin{theorem}\label{theorem:stablecutsetmodlinforest:mktwohardness}
    \stablecutsetmodlinforest is \mktwo-hard under PPTs.
\end{theorem}

\begin{proof}
    We give a polynomial parameter transformation from \setsplittingn to \stablecutsetmodlinforest. 
    
    Let $(U,\F)$ be an instance of \setsplittingn with $n=|U|$ and $m=|\F|$. If some set $F\in\F$ has size at most one then the instance is trivially no and we return a dummy no-instance of \stablecutsetmodlinforest. Otherwise, we construct an instance $(G,X)$ as follows (see also Fig.~\ref{fig:example}):
    \begin{itemize}
        \item We begin by adding five vertices $s$, $a_1$, $a_2$, $b_1$, and $b_2$, called \emph{special vertices} along with edges $\{s,a_1\}$, $\{s,a_2\}$, $\{s,b_1\}$, $\{s,b_2\}$, $\{a_1,a_2\}$, and $\{b_1,b_2\}$.
        \item For each element $u\in U$ we add two vertices $u_a$ and $u_b$, called \emph{element vertices}, and the edge $\{u_a,u_b\}$ to $G$.
        \item For each set $F\in\F$ we add the following two paths of length $2|F|-1$ each, called \emph{set paths} to $G$: A path $L_F=(\ell_{F,1},\ell'_{F,1},\ell_{F,2},\ell'_{F,2},\ldots,\ell_{F,|F|},\ell'_{F,|F|})$ and a path $R_F=(r_{F,1},r'_{F,1},r_{F,2},r'_{F,2},\ldots,r_{F,|F|},r'_{F,|F|})$.
        \item Between special vertices and element vertices we add the following edges:
        \begin{itemize}
            \item We make $a_1$ and $a_2$ adjacent to each vertex $u_a$ with $u\in U$.
            \item We make $b_1$ and $b_2$ adjacent to each vertex $u_b$ with $u\in U$.
        \end{itemize}
        \item Between special vertices and set paths we add the following edges:
        \begin{itemize}
            \item We make $s$ adjacent to each vertex $\ell'_{F,i}$ with $F\in\F$ and $1\leq i\leq |F|$.
            \item We make $s$ adjacent to each vertex $r'_{F,i}$ with $F\in\F$ and $1\leq i\leq |F|$.
            \item We make $a_1$ and $a_2$ adjacent to the final vertex of each set path $R_F$, i.e., to $r'_{F,|F|}$ for each $F\in\F$.
            \item We make $b_1$ and $b_2$ adjacent to the final vertex of each set path $L_F$, i.e., to $\ell'_{F,|F|}$ for each $F\in\F$.
        \end{itemize}
        \item Between element vertices and set paths we add the following edges:
        \begin{itemize}
            \item For each set $F\in\F$ we arbitrarily enumerate its elements as $c_{F,1},\ldots,c_{F,|F|}$, e.g., by following any linear ordering of $U$.
            \item For each $u\in F$ with $u=c_{F,i}$ we make $u_a$ adjacent to $r_{F,i}$ and $r'_{F,i}$ in set path $R_F$, and $u_b$ adjacent to $\ell_{F,i}$ and $\ell'_{F,i}$ in set path $L_F$.
            \item For the first and second element of each set $F$, according to the above enumeration, there are additional edges between element vertices and set paths (recall also that each set has size at least two):
            \begin{itemize}
                \item For $u=c_{F,1}$ we make $u_a$ adjacent to $\ell_{F,2}$ in set path $L_F$, and we make $u_b$ adjacent to $r_{F,2}$ in set path $R_F$.
                \item For $u=c_{F,2}$ we make $u_a$ adjacent to $r_{F,1}$ in set path $R_F$, and we make $u_b$ adjacent to $\ell_{F,1}$ in set path $L_F$.
            \end{itemize}
        \end{itemize}
        \item We let $X=\{s,a_1,a_2,b_1,b_2\}\cup\{u_a,u_b\mid u\in U\}$, i.e., $X$ contains all special and element vertices. Clearly $G-X$ is a disjoint union of all the set paths, i.e., a linear forest.
    \end{itemize}

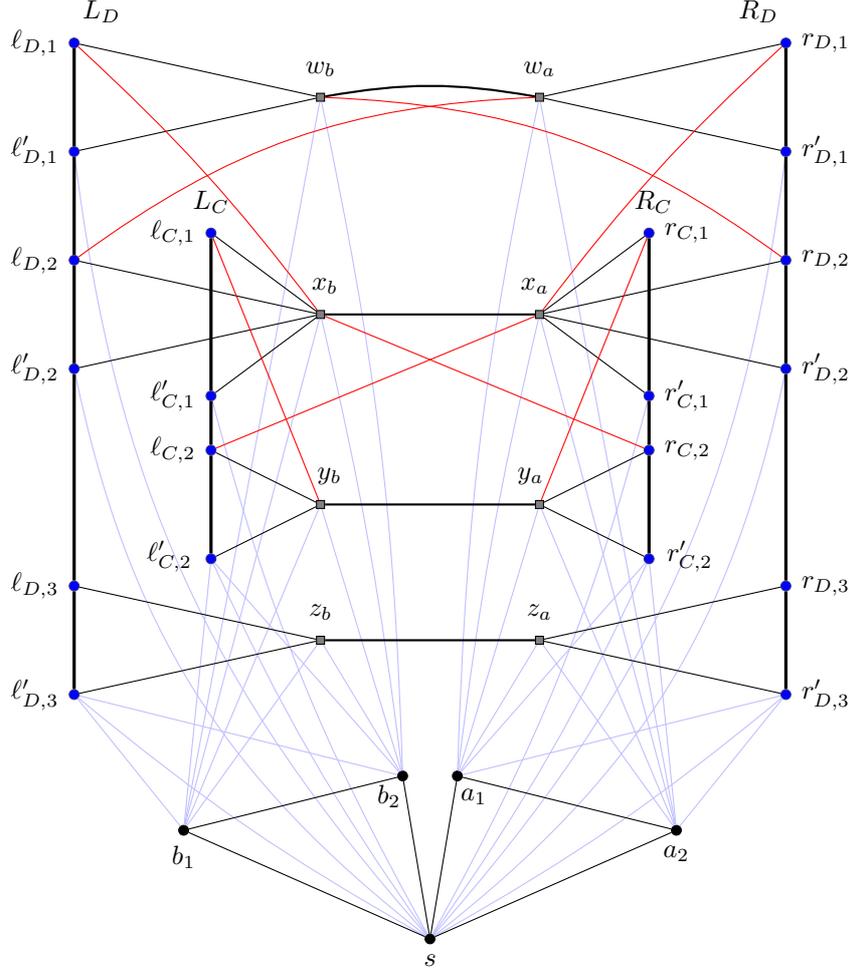
\begin{figure}
\begin{center}
\tikzstyle{vertex}=[draw,circle,inner sep=1.3pt,fill=black] 
\tikzstyle{clause}=[draw,circle,inner sep=1.3pt,fill=black]
\tikzstyle{var}=[draw,rectangle,inner sep=1.5pt,fill=gray] 
\tikzstyle{subsC}=[draw=Blue,circle,inner sep=1.3pt,fill=blue] 
\tikzstyle{subsD}=[draw=Blue,circle,inner sep=1.3pt,fill=blue] 
\tikzstyle{vertexR}=[draw=black!15!flax,circle,inner sep=1.2pt,fill=flax] 
\tikzstyle{vertexB}=[draw=black!15!teal,circle,inner sep=1.2pt,fill=teal] 

\begin{tikzpicture}[scale=.36] 
\node[var] (wa) at (21,36) {}; 
\node[var] (wb) at (13,36) {}; 
\node[var] (xa) at (21,28) {}; 
\node[var] (xb) at (13,28) {}; 
\node[var] (ya) at (21,21) {}; 
\node[var] (yb) at (13,21) {}; 
\node[var] (za) at (21,16) {}; 
\node[var] (zb) at (13,16) {}; 

\node  at (9,31) [label=above:{\small $L_C$}] {};
\node  at (25,31) [label=above:{\small $~R_C$}] {};
\node[subsC] (lc1) at (9,31)  [label=left:{\small $\ell_{C,1}$}] {};
\node[subsC] (rc1) at (25,31)  [label=right:{\small $r_{C,1}$}] {};
\node[subsC] (l'c1) at (9,25) [label=left:{\small $\ell'_{C,1}$}] {};
\node[subsC] (r'c1) at (25,25) [label=right:{\small $r'_{C,1}$}] {};
\node[subsC] (lc2) at (9,23) [label=left:{\small $\ell_{C,2}$}] {};
\node[subsC] (rc2) at (25,23) [label=right:{\small $r_{C,2}$}] {};
\node[subsC] (l'c2) at (9,19) {}; 
\node[subsC] (r'c2) at (25,19) {}; 

\node  at (5,38) [label=above:{\small $L_D$}] {};
\node  at (29,38) [label=above:{\small $R_D$}] {};
\node[subsD] (ld1) at (4,38) [label=left:{\small $\ell_{D,1}$}] {}; 
\node[subsD] (rd1) at (30,38) [label=right:{\small $r_{D,1}$}] {}; 
\node[subsD] (l'd1) at (4,34) [label=left:{\small $\ell'_{D,1}$}] {};
\node[subsD] (r'd1) at (30,34) [label=right:{\small $r'_{D,1}$}] {};
\node[subsD] (ld2) at (4,30) [label=left:{\small $\ell_{D,2}$}] {};
\node[subsD] (rd2) at (30,30) [label=right:{\small $r_{D,2}$}] {};
\node[subsD] (l'd2) at (4,26) [label=left:{\small $\ell'_{D,2}$}] {};
\node[subsD] (r'd2) at (30,26) [label=right:{\small $r'_{D,2}$}] {};
\node[subsD] (ld3) at (4,18) [label=left:{\small $\ell_{D,3}$}] {};
\node[subsD] (rd3) at (30,18) [label=right:{\small $r_{D,3}$}] {};
\node[subsD] (l'd3) at (4,14) [label=left:{\small $\ell'_{D,3}$}] {};
\node[subsD] (r'd3) at (30,14) [label=right:{\small $r'_{D,3}$}] {};

\node[vertex] (s) at (17,5) [label=below:{\small $s$}] {};
\draw[lightblue,thin] (s)[bend angle=7, bend left]to (l'c1);
\draw[lightblue,thin] (s)[bend angle=5, bend left]to (l'c2);
\draw[lightblue,thin] (s)[bend angle=7, bend right]to (r'c1);
\draw[lightblue,thin] (s)[bend angle=5, bend right]to (r'c2);
\draw[lightblue,thin] (s)[bend angle=17, bend left] to (l'd1);
\draw[lightblue,thin] (s)[bend angle=20, bend left]to (l'd2);
\draw[lightblue,thin] (s)[bend angle=5, bend left]to  (l'd3);
\draw[lightblue,thin] (s)[bend angle=17, bend right]to (r'd1);
\draw[lightblue,thin] (s)[bend angle=20, bend right]to (r'd2); 
\draw[lightblue,thin] (s)[bend angle=5, bend right]to (r'd3); 
\node  at (9.05,19.1) [label=left:{\small $\ell'_{C,2}$}] {};
\node at (24.9,19.1) [label=right:{\small $r'_{C,2}$}] {};

\node[vertex] (a1) at (18,11) {}; 
\node at (18.6,10.25)  {\small $a_1$};
\node[vertex] (a2) at (26,9) [label=below:{\small $a_2$}] {};
\draw (a1)--(s)--(a2)--(a1);
\draw[lightblue,thin] (a1)[bend angle=5, bend left] to (wa); 
\draw[lightblue,thin] (a1)[bend angle=3, bend left] to (xa); 
\draw[lightblue,thin] (a1)--(ya); 
\draw[lightblue,thin] (a1)--(za); 
\draw[lightblue,thin] (a1)--(r'c2); 
\draw[lightblue,thin] (a1)--(r'd3); 
\draw[lightblue,thin] (a2)--(wa); 
\draw[lightblue,thin] (a2)--(xa); 
\draw[lightblue,thin] (a2)--(ya); 
\draw[lightblue,thin] (a2)--(za); 
\draw[lightblue,thin] (a2)--(r'c2); 
\draw[lightblue,thin] (a2)--(r'd3);

\node[vertex] (b1) at (8,9) [label=below:{\small $b_1$}] {};
\node[vertex] (b2) at (16,11) {};
\node at (15.5,10.25) {\small $b_2$};
\draw (b1)--(s)--(b2)--(b1);
\draw[lightblue,thin] (b1)--(wb); 
\draw[lightblue,thin] (b1)--(xb); 
\draw[lightblue,thin] (b1)--(yb); 
\draw[lightblue,thin] (b1)--(zb); 
\draw[lightblue,thin] (b1)--(l'c2); 
\draw[lightblue,thin] (b1)--(l'd3);
\draw[lightblue,thin] (b2)[bend angle=5, bend right]to (wb); 
\draw[lightblue,thin] (b2)[bend angle=3, bend right]to (xb); 
\draw[lightblue,thin] (b2)--(yb); 
\draw[lightblue,thin] (b2)--(zb); 
\draw[lightblue,thin] (b2)--(l'c2); 
\draw[lightblue,thin] (b2)--(l'd3);

\node at (21,36) [label=above:{\small $w_a$}] {};
\node at (13,36) [label=above:{\small $w_b$}] {};
\node at (21,28) [label=above:{\small $x_a~$}] {};
\node at (13,28) [label=above:{\small $~x_b$}] {};
\node at (21,21) [label=above:{\small $y_a~~$}] {};
\node at (13,21) [label=above:{\small $~~y_b$}] {};
\node at (21,16) [label=above:{\small $z_a$}] {};
\node at (13,16) [label=above:{\small $z_b$}] {};

\draw[thick] (wa)[bend angle=10, bend right] to(wb);
\draw[thick] (xa)--(xb);
\draw[thick] (ya)--(yb);
\draw[thick] (za)--(zb);

\draw[very thick] (lc1)--(l'c1)--(lc2)--(l'c2);
\draw[very thick] (rc1)--(r'c1)--(rc2)--(r'c2);
\draw (lc1)--(xb)--(l'c1); \draw (rc1)--(xa)--(r'c1);
\draw (lc2)--(yb)--(l'c2); \draw (rc2)--(ya)--(r'c2);
\draw[very thick] (ld1)--(l'd1)--(ld2)--(l'd2)--(ld3)--(l'd3);
\draw[very thick] (rd1)--(r'd1)--(rd2)--(r'd2)--(rd3)--(r'd3);
\draw (ld1)--(wb)--(l'd1); \draw (rd1)--(wa)--(r'd1); 
\draw (ld2)--(xb)--(l'd2); \draw (rd2)--(xa)--(r'd2); 
\draw (ld3)--(zb)--(l'd3); \draw (rd3)--(za)--(r'd3); 

\draw[red] (lc1)--(yb); \draw[red] (rc1)--(ya);
\draw[red] (lc2)--(xa); \draw[red] (rc2)--(xb);
\draw[red] (ld1)[bend angle=5, bend left]to (xb); \draw[red] (rd1)[bend angle=5, bend right]to(xa);
\draw[red] (ld2)[bend angle=17, bend left] to (wa); \draw[red] (rd2)[bend angle=17, bend right] to(wb);
\end{tikzpicture} 
\caption{The graph $G$ from the set splitting instance $(U,{\cal F})$ with $U=\{w,x,y,z\}$, ${\cal F}=\{C,D\}$, $C=\{x,y\}$ and $D=\{w,x,z\}$.} 
\label{fig:example}
\end{center}
\end{figure}
    
    Clearly this construction can be carried out in polynomial time and the set $X$ is a modulator to the class of linear forests. The parameter value $|X|=5+2n$ is polynomial in the input parameter $n$. It remains to verify correctness, i.e., that $(U,\F)$ is yes for \setsplittingn if and only if $(G,X)$ is yes for \stablecutsetmodlinforest.

    Assume first that $(G,X)$ is yes for \stablecutsetmodlinforest and let $S\subseteq V(G)$ be a stable cutset of $G$. Since $a_1$ and $a_2$ as well as $b_1$ and $b_2$ are twins, we may assume that $a_1\notin S$ and $b_1\notin S$. Our first goal is to show that for each vertex $v\in V(G)\setminus S$ there is a path from $a_1$ or $b_1$ to $v$ in $G-S$.

    We say that a vertex $v$ is \emph{weakly reachable} if there is a path from $a_1$ or $b_1$ to $v$ in $G$ whose vertices lie in $V(G)\setminus S$ except that $v\in S$ is permitted. Equivalently, it is required that there exists a path from $a_1$ or $b_1$ to any neighbor $w$ of $v$ in $G-S$. (Using this notion avoids several case distinctions about vertices in $S$.)

    \begin{observation}\label{observation:weaklyreachable:reachable}
        For each $v\in V(G)\setminus S$, weak reachability in $G$ and (regular) reachability are equivalent.
    \end{observation}

    \begin{observation}\label{observation:weaklyreachable:neighbor}
        If $v\notin S$ is (weakly) reachable then each neighbor $w$ of $v$ is weakly reachable (since $w$ has $v$ as reachable neighbor).
    \end{observation}

    \begin{observation}\label{observation:weaklyreachable:triangle}
        If $u$, $v$, and $w$ form a triangle in $G$ and both $u$ and $v$ are weakly reachable then $w$ is weakly reachable (since only one of $u$ and $v$ can be in $S$).
    \end{observation}

    \begin{claim}
        Each vertex of $G$ is weakly reachable.
    \end{claim}

    \begin{proof}
        Clearly, $a_1$ and $b_1$ are weakly reachable. Since $a_1,b_1\notin S$, their neighbors $s$, $a_2$, and $b_2$ are weakly reachable by Observation~\ref{observation:weaklyreachable:neighbor}. Similarly, all element vertices, $u_a$ and $u_b$ for $u\in U$, are weakly reachable as each is adjacent (amongst others) to $a_1$ or $b_1$, and $a_1,b_1\notin S$. It remains to verify that the vertices of the set paths are weakly reachable.

        For this we need a case distinction on whether $s$ is contained in $S$. Assume first that $s\notin S$. Then all vertices $\ell'_{F,i}$ and $r'_{F,i}$, for $F\in\F$ and $i\in[|F|]$, are weakly reachable because they are adjacent to $s$ (Observation~\ref{observation:weaklyreachable:neighbor}). It remains to check weak reachability of vertices $\ell_{F,i}$ and $r_{F,i}$ for $F\in\F$ and $i\in[|F|]$: Let $u\in U$ with $u=c_{F,i}$ (from enumeration of elements of $F$). Then $u_b$, $\ell_{F,i}$, and $\ell'_{F,i}$ form a triangle and we already know that $u_b$ and $\ell'_{F,i}$ are weakly reachable. It follows from Observation~\ref{observation:weaklyreachable:triangle} that $\ell_{F,i}$ is weakly reachable too. Similarly, $u_a$, $r_{F,i}$, and $r'_{F,i}$ form a triangle and $u_a$ as well as $r'_{F,i}$ are weakly reachable, which implies that $r_{F,i}$ is weakly reachable too. This completes the first case.

        In the second case, we assume that $s \in S$ and it follows that the adjacent vertices $\ell'_{F,i}$ and $r'_{F,i}$, for $F\in\F$ and $i\in[|F|]$, are not in $S$ because $S$ is a stable set. Fix any set $F\in\F$; we show that the vertices of its two set paths $L_F$ and $R_F$ are all weakly reachable. By symmetry it suffices to show this for all vertices of path $L_F$. We begin with $\ell_{F,1}$, $\ell'_{F_1}$, $\ell_{F,2}$, and $\ell'_{F,2}$ using the special edges related to the first and second elements of $F$. Let $u,v\in U$ with $u=c_{F,1}$ and $v=c_{F,2}$, and make a case distinction on whether $v_b$ is contained in $S$:
        \begin{itemize}
            \item If $v_b\notin S$ then its neighbors $\ell_{F,1}$, $\ell_{F,2}$, and $\ell'_{F,2}$ are all weakly reachable (Observation~\ref{observation:weaklyreachable:neighbor}). Vertex $\ell'_{F,1}$ is then weakly reachable because it forms a triangle with the weakly reachable vertices $u_b$ and $\ell_{F,1}$.
            \item Else, if $v_b\in S$, then its neighbors $\ell_{F,1}$ and $\ell_{F,2}$ are not in the stable set $S$, along with $\ell'_{F,1}$ and $\ell'_{F,2}$ as observed at the start of the paragraph. Since these vertices form a path, they are either all weakly reachable or none of them is (Observation~\ref{observation:weaklyreachable:neighbor}). We observe that $u_a\notin S$ or $u_b\notin S$ must hold (as they are adjacent), so neighbors $L_{F,2}$ or $L_{F,1}$ are weakly reachable; this implies that all of $\ell_{F,1}$, $\ell'_{F_1}$, $\ell_{F,2}$, and $\ell'_{F,2}$ are weakly reachable.
        \end{itemize}
        For the remaining vertices of the set path $L_F$ we can use an inductive argument: We show that $\ell_{F,k}$ and $\ell'_{F,k}$ are weakly reachable for $k\in[|F|]$, having just proved this for $k=1,2$. For $k\geq 3$ we can, in particular, assume that $\ell'_{F,k-1}$ is weakly reachable. Since $\ell'_{F,k-1}\notin S$, it follows that its neighbor $\ell_{F,k}$ is weakly reachable. Then $\ell'_{F,k}$ is weakly reachable too as its forms a triangle with the weakly reachable vertices $\ell_{F,k}$ and $u_b$ with $u_b=c_{F,k}$. This completes the proof of the claim.
    \end{proof}

    It follows from Observation~\ref{observation:weaklyreachable:reachable} that each vertex $v\in V(G)\setminus S$ is reachable, i.e., that there is a path from $a_1$ or $b_1$ to $v$ in $G-S$ (or both). It follows directly that $G-S$ has at most two connected components. Since $S$ is a stable cutset, it follows that $G-S$ has exactly two connected components, one containing $a_1$ and the other containing $b_1$. This in turn implies that $s\in S$ as it is adjacent to both vertices in $G$. Moreover, for each element $u\in U$ exactly one of $u_a$ and $u_b$ must be in $S$: (1) We have the path $(a_1,u_a,u_b,b_1)$ in $G$ and $a_1,b_1\notin S$, hence at least one of them is in $S$. (2) They are adjacent, hence at most one of them is in the stable set $S$.

    \begin{claim}
        For each set $F\in\F$ there are elements $u,v\in F$ with $u_a\in S$ and $v_a\notin S$.
    \end{claim}

    \begin{proof}
        Assume for contradiction that $u_a\in S$ for all $u\in F$. (The case that $u_a\notin S$ for all $u\in F$ is symmetric because then $u_b\in S$ for all $u\in F$, as exactly one of $u_a$ and $u_b$ is in $S$.) Then $r_{F,i},r'_{F,i}\notin S$ because each of them is adjacent to $u_a$ for $u=c_{F,i}$ (from enumeration of elements $u\in F$). Let $v=c_{F,1}$ and recall that $v_b$ is adjacent to $r_{F,2}$ . Since $v_a\in S$ we have $v_b\notin S$, which implies that
        \[
        (b_1,v_b,r_{F,2},r'_{F,2},r_{F,3},r'_{F,3},\ldots,r_{F,|F|},r'_{F,|F|},a_1)
        \]
        is a path in $G-S$ that connects $a_1$ and $b_2$; a contradiction. This completes the proof of the claim.
    \end{proof}

    We can now construct a splitting set for $(U,\F)$: Let $S'$ contain $u$ if $u_a\in S$. The previous claim implies that for each set $F\in\F$ there is an element $u$ with $u_a\in S$, hence $u\in S'$, and an element $v$ with $v_a\notin S$, hence $v\notin S'$. Thus $S'$ splits $F$, implying that $(U,\F)$ is yes for \setsplittingn.

    Now assume that $(U,\F)$ is yes for \setsplittingn and let $S'\subseteq U$ be a set that splits all $F\in\F$, i.e., such that $F\nsubseteq S'$ and $F\nsubseteq U\setminus S'$ for all $F\in\F$. We will construct a partition of $V(G)$ into three sets $S$, $A$, and $B$ such that the following three properties hold (1) $S$ is a stable set, (2) $A$ and $B$ are nonempty, and (3) there are no edges with one endpoint in $A$ and the other in $B$: 
    \begin{itemize}
        \item We begin with $S=\{s\}$, $A=\{a_1,a_2\}$, and $B=\{b_1,b_2\}$. We will add the remaining vertices of $G$ until we have a partition while maintaining properties (1)-(3). Clearly all three properties hold at the moment and adding further vertices to the three sets cannot break property (2).
        \item For each $u\in U$,
        \begin{itemize}
            \item if $u\in S'$ then we add $u_a$ to $S$ and $u_b$ to $B$,
            \item while if $u\notin S'$ then we add $u_a$ to $A$ and $u_b$ to $S$.
        \end{itemize}
        Since the only edges between these vertices as well as to the set $\{s,a_1,a_2,b_1,\allowbreak 
        b_2\}$ are of form $\{u_a,u_b\}$, $\{u_a,a_i\}$, and $\{u_b,b_i\}$, this maintains all three properties. (In particular, there are no edges between element vertices corresponding to different elements, and there are no edges from $s$ to any element vertex.) It will be helpful to observe that vertices $u_a$, with $u\in U$, are each in $A$ or $S$, while vertices $u_b$, with $u\in U$, are each in $B$ or $S$.
        \item We will now handle the vertices of the set paths $L_F$ and $R_F$ by treating one set $F\in\F$ at a time. Note that there are no edges between the set paths for different sets $F,F'\in\F$ (nor in fact between the two set paths $L_F$ and $R_F$ for any single $F\in\F$). Hence, to maintain the three properties it suffices to check the adjacency of the vertices on the current set paths to the element vertices and to the special vertices $s$, $a_1$, $a_2$, $b_1$, and $b_2$. For assigning the vertices of $L_F$ and $R_F$ to the sets $S$, $A$, and $B$ we make a case distinction on whether the first element $c_{F,1}$ of the current set $F$ is contained in the splitting set $S'$; the two cases are symmetric.
        \item For each $F\in\F$ with $c_{F,1}\in S'$ let $k\in\N$ be maximal such that $c_{F,1},\ldots,c_{F,k}\in S'$. Since $F\nsubseteq S'$, we have $1\leq k\leq |F|-1$. We now proceed depending on the value of $k$:
        \begin{itemize}
            \item If $k=1$, i.e., if $c_{F,1}\in S'$ and $c_{F,2}\notin S'$ then we add to $S$ the vertex $r_{F,2}$, we add to $A$ all other vertices of the set path $R_F$, and we add to $B$ all vertices of the set path $L_F$. We check separately the new vertices in all three sets $S$, $A$, and $B$ for maintaining properties (1) and (3). Let $u=c_{F,1}$ and $v=c_{F,2}$, and note that $u_a \in S$, $u_b\in B$, $v_a\in A$, and $v_b\in S$ as $c_{F,1}\in S'$ and $c_{F,2}\notin S'$.
            \begin{itemize}
                \item The additional vertex $r_{F,2}$ in $S$ is adjacent to $r'_{F,1}\in A$, to $r'_{F,2}\in A$, to $u_b \in B$, and to $v_a \in A$, preserving property (1).
                \item The additional vertices $r_{F,1},r'_{F,1},r'_{F,2},r_{F,3},r'_{F,3},\ldots,r_{F,|F|},r'_{F,|F|}$ in $A$ have edges amongst themselves and to $r_{F,2}\in S$, to the vertex $s\in S$, to the vertices $a_1,a_2\in A$, and to element vertices $w_a\in A\cup S$. There are no edges to vertices in $B$ (noting once more that there are no edges from $R_F$ to other set paths, also not to vertices on $L_F$).
                \item The additional vertices $\ell_{F,1},\ell'_{F,1},\ell_{F,2},\ell'_{F,2},\ldots,\ell_{F,|F|},\ell'_{F,|F|}$ in $B$ have edges amongst themselves and to the vertex $s\in S$, to the vertices $b_1,b_2\in B$, to the vertex $u_a\in S$, and to element vertices $w_b\in B\cup S$. There are no edges to vertices in $A$, so altogether we also preserve property (3).
            \end{itemize}
            Altogether, adding the vertices of the set paths for $F$ to the sets $S$, $A$, and $B$ preserves all three properties (1)-(3) in this case.
            \item If $2\leq k\leq |F|-1$, i.e., if $c_{F,1},\ldots,c_{F,k}\in S'$ and $c_{F,k+1}\notin S'$ then we add to $S$ the vertex $r_{F,k+1}$, add to $A$ the vertices $r'_{F,k+1},r_{F,k+2},r'_{F,k+2},\ldots,\allowbreak r_{F,|F|},r'_{F,|F|}$ (all those of $R_F$ after $r_{F,k+1}$), and add to $B$ the remaining vertices of $R_F$, namely $r_{F,1},r'_{F,1},r_{F,2},r'_{F,2},\ldots,r_{F,k},r'_{F,k}$, along with all vertices of the set path $L_F$. We again check separately the new vertices in all three sets $S$, $A$, and $B$ for maintaining properties (1) and (3). Let $u=c_{F,1}$ and $w=c_{F,k+1}$, and note that $u_a\in S$, $u_b\in B$, $w_a\in A$, and $w_b\in S$ since $u\in S'$ and $w\notin S'$.
            \begin{itemize}
                \item The additional vertex $r_{F,k+1}$ in $S$ is adjacent to $r'_{F,k}\in B$, to $r'_{F,k+1}\in A$, and to $w_a\in A$, preserving property (1).
                \item The additional vertices $r'_{F,k+1},r_{F,k+2},r'_{F,k+2},\ldots,r_{F,|F|},r'_{F,|F|}$ in $A$ have edges amongst themselves and to $r_{F,k+1}\in S$, to the vertex $s\in S$, to the vertices $a_1,a_2\in A$, and to element vertices $x_a\in A\cup S$. There are no edges to vertices in $B$.
                \item The additional vertices $r_{F,1},r'_{F,1},r_{F,2},r'_{F,2},\ldots,r_{F,k},r'_{F,k}$ in $B$ have edges amongst themselves and to $r_{F,k+1}\in S$, to the vertex $s\in S$, (not to $a_1$ and $a_2$ as $k\leq |F|-1$,), to the vertex $u_b\in B$ (from $r_{F,2}$ as $k\geq 2$), and to element vertices $v_a \in S$ for $v=c_{F,1},\ldots,c_{F,k}\in S'$. There are no edges to vertices in $A$.
                \item The additional vertices $\ell_{F,1},\ell'_{F,1},\ell_{F,2},\ell'_{F,2},\ldots,\ell_{F,|F|},\ell'_{F,|F|}$ in $B$ have edges amongst themselves and to the vertex $s\in S$, to the vertices $b_1,b_2\in B$, to the vertex $u_a\in S$, and to element vertices $x_b\in B\cup S$. There are no edges to vertices in $A$, so altogether we also preserve property (3).
            \end{itemize}
            Altogether, also in this case we get that adding the vertices of the set paths for $F$ to the sets $S$, $A$, and $B$ maintains all three properties (1)-(3).
        \end{itemize}
        For $F\in\F$ with $c_{F,1}\in S'$ we have showed how to add the vertices of set paths for $F$ to $S$, $A$, and $B$ while maintaining all three properties.
        \item For each $F\in\F$ with $c_{F,1}\notin S'$ let $k\in\N$ be maximal such that $c_{F,1},\ldots,c_{F,k}\notin S'$. Since $F\nsubseteq U\setminus S'$, we have $1\leq k\leq |F|-1$. We again proceed depending on the value of $k$; since the case is symmetric we only specify the placement of the vertices but do not go over the preservation of properties again:
        \begin{itemize}
            \item If $k=1$, i.e., if $c_{F,1}\notin S'$ and $c_{F,2}\in S'$ then we add to $S$ the vertex $\ell_{F,2}$, we add to $B$ all other vertices of the set path $L_F$, and we add to $A$ all vertices of the set path $R_F$.
            \item If $2\leq k \leq |F|-1$, i.e., if $c_{F,1},\ldots,c_{F,k}\notin S'$ and $c_{F,k+1}\in S'$ then we add to $S$ the vertex $\ell_{F,k+1}$, add to $B$ the vertices $\ell'_{F,k+1},\ell_{F,k+2},\ell'_{F,k+2},\ldots,\allowbreak \ell_{F,|F|},\ell'_{F,|F|}$ (all those of $L_F$ after $\ell_{F,k+1}$), and add to $A$ the remaining vertices of $L_F$, namely $\ell_{F,1},\ell'_{F,1},\ell_{F,2},\ell'_{F,2},\ldots,\ell_{F,k},\ell'_{F,k}$, along with all vertices of the set path $R_F$.
        \end{itemize}
        It can be checked, analogously to the previous case with $c_{F,1}\in S'$, that in both subcases all properties (1)-(3) are maintained.
    \end{itemize}
    Overall, we obtain a partition of $V(G)$ into $S$, $A$, and $B$ such that (1) $S$ is a stable set, (2) $A$ and $B$ are nonempty, and (3) no edge of $G$ has one endpoint in $A$ and the other in $B$. It follows directly that $S$ is a stable cutset in $G$. Hence, $(G,X)$ is yes. This completeness the correctness proof for our PPT from \setsplittingn to \stablecutsetmodlinforest. Since \setsplitting is complete for \mktwo under PPTs, it follows that \stablecutsetmodlinforest is hard for \mktwo under PPTs, as claimed.
\end{proof}

As a free corollary, we get the same hardness result for \stablecutset parameterized by the size of a given (connected) dominating set. It suffices to observe that the set $X$ in the reduction from \setsplittingn to \stablecutsetmodlinforest is also a (connected) dominating set of size $5+2n$.

\begin{corollary}
    \stablecutset parameterized by the size of a (connected) dominating set is hard for \mktwo under PPTs.
\end{corollary}

By a small modification of the reduction, one can ensure that $G-X$ is a single path, i.e., that $X$ is a modulator to the class of all paths, a subclass of the linear forests (so we get a slightly stricter parameter). This is similar in spirit to a lower bound for \threecoloring parameterized by distance from a single path~\cite{JansenK13}.

\begin{corollary}
    \stablecutsetmodpath is hard for \mktwo under PPTs.
\end{corollary}

\begin{proof}
    The idea is to connect the linear forest $G-X$ into a single path. The path components are appended to one another through new vertices that are forced to be in $S$. We formalize the additional construction to get from $(G,X)$ to $(G',X)$ and then briefly sketch correctness.

    Given any instance $(U,\F)$ of \setsplittingn, we first construct an equivalent instance $(G,X)$ of \stablecutsetmodlinforest as in the proof of Theorem~\ref{theorem:stablecutsetmodlinforest:mktwohardness}. Then, using an arbitrary enumeration of sets in $\F$, say, $\F=\{F_1,\ldots,F_m\}$, we add additional vertices to $G$ in order to obtain $G'$ and hence the instance $(G',X)$:
    \begin{itemize}
        \item Take $m$ new vertices $c_1,\ldots,c_m$ and make each $c_i$ adjacent to $\ell'_{F_i,|F_i|}$ and $r_{F_i,1}$ as well as to $a_1$, $a_2$, $b_1$, and $b_2$.
        \item Take $m-1$ new vertices $d_1,\ldots,d_{m-1}$ and make each $d_i$ adjacent to $r'_{F_i,|F_i|}$ and $\ell_{F_{i+1},1}$ as well as to $a_1$, $a_2$, $b_1$, and $b_2$.
    \end{itemize}
    This completes the (extended) construction. Observe that $G'-X$ consists of a single path of form 
    \[
        L_{F_1},c_1,R_{F_1},d_1,L_{F_2},\ldots,R_{F_{m-1}},d_{m-1},L_{F_m},c_m,R_{F_m}.
    \]
    Thus, $X$ is indeed a modulator of $G'$ to a single path, with size polynomial in $n$ as before, and it remains to check correctness. It will be convenient to have $T:=\{c_1,\ldots,c_m,d_1,\ldots,d_{m-1}\}$.

    Assume first that $(U,\F)$ is yes for \setsplittingn. We showed that this implies the existence of a stable cutset $S$ for $G$ and consistent partition $V(G)\setminus S=A\dunion B$ with the following properties:
    \begin{itemize}
        \item We have $a_1,a_2\in A$ and $b_1,b_2\in B$, i.e., $a_1,a_2,b_1,b_2\notin S$.
        \item We have $\ell_{F,1},r_{F,1}\notin S$ for all $F\in\F$.
        \item We have $\ell'_{F,|F|},r'_{F,|F|}\notin S$ for all $F\in \F$.
    \end{itemize}
    Thus, none of the vertices adjacent to the new vertices in $T$ are in $S$. Since we also have that $T$ is a stable set, it follows that $S'=S\cup T$ is a stable cutset with consistent partition $V(G')\setminus S=V(G)\setminus S=A\dunion B$. Thus, $(G',X)$ is yes for \stablecutsetmodpath.

    Now assume that $(G',X)$ is yes for \stablecutsetmodpath and let $S'$ be a stable cutset of $G'$. Let $S=S'\cap V(G)$, which is a stable set in $G$ (though not necessarily a cutset of $G$). Assume for contradiction that $G-S$ is connected and let $C=V(G)\setminus S$ be the sole connected component of $G-S$. Recall that the vertices of $T$ are each adjacent to $a_1$ and $a_2$, out of which at most one, say $a_2$, is in $S$. But then each vertex of $T$ is either in $S'$ or it is adjacent to $a_1\in C$, which implies that $G'-S'$ is connected too (using that $G=G'-T$); a contradiction. It follows that $G-S$ is disconnected and that $S$ is a stable cutset of $G$. But then $(G,X)$ is yes for \stablecutsetmodlinforest and we already know that this implies $(U,\F)$ to be yes for \setsplittingn.
\end{proof}

\subsection{Lower bounds for other parameterizations}

We show a simple construction that given $t$ graphs $G_1,\ldots,G_t$ returns a single graph $G$ such that $G$ has a stable cutset if and only if at least one graph $G_i$ has a stable cutset. This construction is somewhat similar to a construction in obtaining lower bounds for Matching Cut in~\cite{KomusiewiczKL20}. 
Using the standard framework for kernelization lower bounds via cross-compositions~\cite{BodlaenderJK14} then allows to rule out polynomial kernelizations for \stablecutset parameterized by either treewidth or solution size unless \containment.

\introduceparameterizedproblem{\stablecutsetk}{A graph $G$ and $k\in\N$.}{$k$.}{Does $G$ have a stable cutset of size at most $k$?}

\introduceparameterizedproblem{\stablecutsettw}{A graph $G$ and a tree decomposition $\T$ of $G$.}{The width of $\T$.}{Does $G$ have a stable cutset?}

\begin{lemma}
    Given $t$ graphs $G_1,\ldots,G_t$, one can efficiently construct a graph $G$ such that $G$ has a stable cutset if and only if at least one graph $G_i$ has a stable cutset. Moreover, the minimum size of stable cutset in $G$ is at most the minimum size of any stable cutset of any graph $G_i$, and the treewidth of $G$ is at most the maximum treewidth of any graph $G_i$ plus an additive constant.
\end{lemma}

\begin{proof}
    W.l.o.g., each of the graphs $G_i$ has at least one edge. (If any graph $G_i$ is disconnected then return $G:=G_i$, clearly fulfilling the lemma. We may also discard any graphs $G_i$ that are cliques, including the empty graph and the graph with a single vertex.) Take the disjoint union of all the graphs and choose an arbitrary edge $\{u_i,v_i\}$ in each graph $G_i$. Add two adjacent vertices $p$ and $q$, and make them fully adjacent to all vertices $u_i$ and $v_i$. Return the obtained graph $G$. Clearly, its treewidth is at most the maximum of the treewidths of the graph $G_i$ and three (due to the cliques of size four on $p$, $q$, and any $u_i$ and $v_i$). It remains to check the existence of stable cutsets.

    We show that any stable cutset $S\subseteq V(G_i)$ is also a stable cutset of $G$, which also proves the remaining moreover part of the lemma. (The converse of the main statement is proved below.) Let $V(G_i)\setminus S=A\dunion B$ be a consistent partition for $S$ in $G_i$. Clearly, $u_i$ and $v_i$ cannot be one in $A$ and one in $B$, as they are adjacent. By symmetry, assume that $u_i,v_i\in S\cup A$. It follows directly that $(A\cup\{p,q\}\cup \bigcup_{j\neq i} V(G_j))\dunion B$ is a consistent partition for $S$ in $G$.
    
    For the converse, let $S\subseteq V(G)$ be a stable cutset of $G$. 
    Then at least one of $p$ and $q$, and for each $i$, at least one of $u_i$ and $v_i$ are not in $S$. Hence $S\cap V(G_i)$ is a stable cutset of $G_i$ for some $i$ (otherwise $G-S$ would be connected).
\end{proof}

\begin{theorem}
    \stablecutsetk and \stablecutsettw have no polynomial kernelization unless \containment.
\end{theorem}

\begin{observation}
    The lower bound extends also to parameterization by pathwidth, tree-depth, and vertex integrity (but clearly not by size of largest connected component, which has a trivial kernelization).
\end{observation}

\section{Conclusion}\label{section:conclusion}

We have initiated the study of polynomial kernelizations for \stablecutset. Our most general positive results are the kernelizations relative to the size of modulators to the classes of cluster respectively co-cluster graphs. Complementing this, the lower bound relative to size of a modulator to a single path rules out positive results for many (if not most) often-studied graph classes above cluster and/or co-cluster graphs. An interesting case to study may be the size of a modulator to the class of cographs. It would also be interesting to prove lower bounds on the size of kernelizations for the positive cases that we obtained, and to push towards making these bounds tight.

\subsubsection*{Acknowledgements.}
VBL thanks Hoang-Oanh Le for discussions on stable cutsets, and beyond.

\bibliographystyle{abbrv}
\bibliography{scs.bib}

\end{document}